\documentclass{llncs}
\usepackage[usenames]{color}
\usepackage{times}
\usepackage{bm}
\usepackage{amsmath}
\usepackage{thm-restate}
\usepackage{amssymb}
\usepackage{paralist}
\usepackage{wrapfig}
\usepackage{stackrel}
\usepackage{adjustbox}
\usepackage{booktabs}
\usepackage{subcaption}
\usepackage[linesnumbered,ruled,noend]{algorithm2e}
\usepackage{skull}

\InputIfFileExists{synctex.tex}{%
\synctex=1
}{%
}

\InputIfFileExists{cutmargins.tex}{%
	\pdfpagesattr{/CropBox [92 112 523 758]} % LNCS page made big 
}{%
}

\newcommand{\td}[1]{\textcolor{blue}{\ifmmode \text{[#1]}\else [#1] \fi}}
\newcommand{\tv}[1]{\textcolor{magenta}{\ifmmode \text{[TV: #1]}\else [TV: #1] \fi}}
\newcommand{\vh}[1]{\textcolor{cyan}{\ifmmode \text{[VH: #1]}\else [VH: #1] \fi}}
\newcommand{\lh}[1]{\textcolor{green}{\ifmmode \text{[LH: #1]}\else [LH: #1] \fi}}
\newcommand{\rem}[1]{\textcolor{red}{\ifmmode \text{[#1]}\else [#1] \fi}}

\newcommand{\bA}[0]{\boldsymbol{A}}
\newcommand{\bE}[0]{\boldsymbol{E}}

\newcommand{\bI}[0]{\boldsymbol{I}}

\newcommand{\real}[0]{\mathbb{R}}
\newcommand{\nat}[0]{\mathbb{N}}

\newcommand{\symdiff}[0]{\mathop{\bigtriangleup}}
\newcommand{\distr}[0]{\mathrm{Pr}}
\newcommand{\distrP}[0]{\distr_P}
\newcommand{\distan}[0]{d}
\newcommand{\distanP}[0]{\distan_P}
\newcommand{\distanPexp}[0]{\distan_{\pexp}}
\newcommand{\lang}[0]{L}
\newcommand{\langof}[1]{\lang(#1)}
\newcommand{\langinof}[2]{\lang_{#1}(#2)}
\newcommand{\blang}[0]{\lang^{\flat}}

\newcommand{\blanginof}[2]{\blang_{#1}(#2)}

\newcommand{\range}[2]{\langle #1, #2\rangle}

\newcommand{\trans}[0]{\delta}
\newcommand{\statesof}[1]{Q[#1]}
\newcommand{\initof}[1]{I[#1]}
\newcommand{\finof}[1]{F[#1]}
\newcommand{\transof}[1]{\trans[#1]}

\newcommand{\pinit}[0]{\boldsymbol{\alpha}}
\newcommand{\pinitof}[1]{\pinit[#1]}
\newcommand{\pfin}[0]{\boldsymbol{\gamma}}
\newcommand{\pfinof}[1]{\pfin[#1]}
\newcommand{\ptrans}[0]{\boldsymbol{\Delta}}
\newcommand{\ptransa}[0]{\ptrans_a}
\newcommand{\ptransaof}[2]{\ptransa[#1, #2]}

\newcommand{\supp}[0]{\mathit{supp}}
\newcommand{\suppof}[1]{\supp(#1)}
\newcommand{\trim}[0]{\mathit{trim}}
\newcommand{\trimof}[1]{\trim(#1)}

\newcommand{\sr}[0]{\rho}
\newcommand{\srof}[1]{\sr(#1)}

\newcommand{\runsto}[0]{\leadsto}
\newcommand{\runstoover}[1]{\stackrel{#1}{\runsto}}

\newcommand{\ltr}[1]{\xrightarrow{#1}} % labelled->

\newcommand{\restr}[2]{#1_{|#2}}    % restriction of automaton: \restr{A}{Q\R}
\newcommand{\reach}[0]{\mathit{reach}}       % reachable states reach(R)
\newcommand{\reachof}[1]{\reach(#1)} 

\newcommand{\pspace}[0]{\textbf{PSPACE}}          % PSPACE
\newcommand{\npspace}[0]{\textbf{NPSPACE}}        % NPSPACE
\newcommand{\ptime}[0]{\textbf{PTIME}}            % PTIME
\newcommand{\nc}[0]{\textbf{NC}}                  % NC
\newcommand{\dunion}[0]{\mathop{\uplus}}          % disjoint union
\newcommand{\onevec}[0]{\boldsymbol{1}}           % unit vector
\newcommand{\vecof}[1]{[#1]}                      % vector

\newcommand{\bigO}[0]{\mathcal{O}}                % complexity O(...)
\newcommand{\bigOof}[1]{\bigO(#1)}

\newcommand{\errfunc}[0]{\mathit{error}}
\newcommand{\errfuncof}[3]{\errfunc(#1,#2,#3)}
\newcommand{\errfuncP}[0]{\errfunc_p}
\newcommand{\errfuncPof}[3]{\errfuncP(#1,#2,#3)}
\newcommand{\errfuncSL}[0]{\errfunc_{\mathit{sl}}}
\newcommand{\errfuncSLof}[3]{\errfuncSL(#1,#2,#3)}

\newcommand{\stlab}[0]{\ell}
\newcommand{\stlabof}[1]{\stlab(#1)}
\newcommand{\stlabP}[0]{\stlab_p}
\newcommand{\stlabSL}[0]{\stlab_{\mathit{sl}}}
\newcommand{\reduce}[0]{\mathit{reduce}}
\newcommand{\reduceof}[2]{\reduce(#1,#2)}
\newcommand{\reduceP}[0]{\reduce_p}
\newcommand{\reducePof}[2]{\reduceP(#1,#2)}
\newcommand{\reduceSL}[0]{\reduce_{\mathit{sl}}}
\newcommand{\reduceSLof}[2]{\reduceSL(#1,#2)}
\newcommand{\labfunc}[0]{\mathit{label}}
\newcommand{\labfuncof}[2]{\labfunc(#1, #2)}
\newcommand{\labfuncP}[0]{\labfunc_p}
\newcommand{\labfuncSL}[0]{\labfunc_{\mathit{sl}}}
\newcommand{\selfloop}[0]{\mathit{sl}}
\newcommand{\selfloopof}[2]{\selfloop(#1,#2)}
\newcommand{\weight}[2]{\mathit{weight}_{#1}(#2)}
\newcommand{\downclosof}[1]{\lfloor #1 \rfloor}
\newcommand{\downclosPof}[1]{\downclosof{#1}_p}
\newcommand{\downclosSLof}[1]{\downclosof{#1}_\mathit{sl}}
\newcommand{\ordP}[0]{\mathrel{\sqsubseteq_p}}
\newcommand{\ordSL}[0]{\mathrel{\sqsubseteq_\mathit{sl}}}

\newcommand{\appreal}[0]{\textsc{Appreal}}
\newcommand{\snort}[0]{\textsc{Snort}}
\newcommand{\netbench}[0]{\textsc{Netbench}}
\newcommand{\alergia}[0]{\textsc{Alergia}}

\newcommand{\reducetool}[0]{\textsc{Reduce}}

\newcommand{\linux}[0]{\textsc{Linux}}
\newcommand{\debian}[0]{\textsc{Debian}}

\newcommand{\cone}[0]{\textbf{C1}}

\newenvironment{condition}[1]{%
\smallskip
\par
\noindent
\textbf{Condition #1}%
\it%
}
{}

\renewenvironment{subparagraph}[1]{%
\smallskip
\par
\noindent
\textbf{#1}
\noindent}
{}

\newcommand{\ordalglab}[0]{\mathrel{\preceq}_{A,\labfuncof A P}}
\newcommand{\ordalglabgrc}[0]{\mathrel{\preceq}_{A,\stlabSL^2}}

\newcommand{\ared}[0]{A^{\text{\textsc{Red}}}}
\newcommand{\aapp}[0]{A^{\text{\textsc{App}}}}
\newcommand{\amal}[0]{A_{\text{\texttt{mal}}}}
\newcommand{\aatt}[0]{A_{\text{\texttt{att}}}}
\newcommand{\abac}[0]{A_{\text{\texttt{bd}}}}
\newcommand{\amalp}[0]{\amal'}
\newcommand{\aattp}[0]{\aatt'}
\newcommand{\abacp}[0]{\abac'}
\newcommand{\amalred}[0]{\amal^{\text{\textsc{Red}}}}
\newcommand{\amalapp}[0]{\amal^{\text{\textsc{App}}}}
\newcommand{\aattred}[0]{\aatt^{\text{\textsc{Red}}}}
\newcommand{\aattapp}[0]{\aatt^{\text{\textsc{App}}}}
\newcommand{\abacred}[0]{\abac^{\text{\textsc{Red}}}}
\newcommand{\abacapp}[0]{\abac^{\text{\textsc{App}}}}

\newcommand{\phttp}[0]{P_{\mathit{HTTP}}}
\newcommand{\timeof}[1]{\mathit{time}(#1)}

\let\mc\multicolumn                      % to allow \mc1l{foo}

\newcommand{\lutsof}[1]{\mathit{LUTs}(#1)}

\newcommand{\url}[1]{\texttt{#1}}

\newcommand{\pprod}[0]{\mathop{\oplus}}

\newcommand{\pexp}[0]{P_{\mathit{Exp}}}
\newcommand{\qnew}[0]{q_{\mathit{new}}}
\newcommand{\expl}[1]{\mbox{\{#1\}}}
\newcommand{\syslin}[0]{\mathcal{S}}
\newcommand{\meq}[0]{M_{\mathit{eq}}}
\newcommand{\msyslin}[0]{M_{\mathit{SysLin}}}
\newcommand{\unkn}[1]{\xi{\scriptstyle[#1]}}

\title{
  Approximate Reduction of Finite Automata for High-Speed Network Intrusion Detection\\ (Technical Report)
\vspace{-3mm}
}

\author{
  Milan \v{C}e\v{s}ka,
  Vojt\v{e}ch Havlena,
  Luk\'{a}\v{s} Hol\'{i}k,
  Ond\v{r}ej Leng\'{a}l, and
  Tom\'{a}\v{s} Vojnar
\vspace{-1mm}
}

\institute{
  FIT, Brno University of Technology, IT4Innovations Centre of Excellence, Czech~Republic
\vspace{-7mm}
}

\begin{document} 
%%%%%%%%%%%%%%%%%%%%%%%%%%%%%%%%%%%%%%%%%%%%%%%%%%%%%%%%%%%%%%%%%%%%%%%%%%%%%%%%

\maketitle

\begin{abstract} We consider the problem of \emph{approximate reduction of
non-de\-ter\-mi\-ni\-stic automata} that appear in hardware-accelerated network
intrusion detection systems (NIDSes).
We define an error \emph{distance} of a reduced automaton from the original one
as the probability of packets being incorrectly classified by the reduced automaton
(wrt the probabilistic distribution of packets in the network traffic).
We use this notion to design an \emph{approximate reduction procedure} that
achieves a~great size reduction (much beyond the state-of-the-art language
preserving techniques) with a~controlled and small error.
We have implemented our approach and evaluated it on use cases from \snort{},
a~popular NIDS.
Our results provide experimental evidence that the method can be highly
efficient in practice, allowing NIDSes to follow the rapid growth in the speed
of~networks. \end{abstract}

%%%%%%%%%%%%%%%%%%%%%%%%%%%%%%%%%%%%%%%%%%%%%%%%%%%%%%%%%%%%%%%%%%%%%%%%%%%%%%%%
\vspace{-9.0mm}
\section{Introduction}\label{sec:label}
\vspace{-3.0mm}
%%%%%%%%%%%%%%%%%%%%%%%%%%%%%%%%%%%%%%%%%%%%%%%%%%%%%%%%%%%%%%%%%%%%%%%%%%%%%%%%

% TV: Start with something for the impatient.
%
% The paper proposes a new approach to reducing the size of non-deterministic
% finite automata (NFAs), which---unlike traditional approaches to automata
% reductions---does not exactly preserve the language of the NFAs but offers
% much better reduction capabilities.
%
% The approach is in particular tuned for the context of dealing with automata
% used in hardware-accelerated network traffic monitoring, which, in turn, is
% motivated by the recent boom in the number of security incidents in computer
% networks.

The recent years have seen a~boom in the number of security incidents in
computer networks.  In order to alleviate the impact of network attacks and
intrusions, Internet  providers want to detect malicious traffic at their
network's entry points and on the backbones between sub-networks.
Software-based network intrusion detection systems (NIDSes), such as the popular
open-source system \snort~\cite{snort}, are capable of detecting suspicious
network traffic by testing (among others) whether a packet payload matches a regular expression (regex)
describing known patterns of malicious traffic. NIDSes collect and maintain vast
databases of such regexes that are typically divided into groups
according to types of the attacks and target protocols.

\emph{Regex matching} is the most computationally demanding task of a~NIDS as
its cost grows with the speed of the network traffic as well as with the number
and complexity of the regexes being matched. The current software-based
NIDSes cannot perform the regex matching on networks beyond
1\,Gbps~\cite{Becchi2009,Korenek2007}, so they cannot handle the current speed of
backbone networks ranging between tens and hundreds of~Gbps. A promising
approach to speed up NIDSes is to (partially) offload regex matching into
hardware~\cite{Korenek2007,Kastil2009,Matousek2016}. The hardware then serves as
a~pre-filter of the network traffic, discarding the majority of the packets from further
processing.  Such pre-filtering can easily reduce the traffic the NIDS needs to
handle by two or three orders of magnitude~\cite{Korenek2007}.

%!!!!!!!!!!!!!!!!!!!!!!!!!!!!!!!!!!!!!!!!!
\enlargethispage{4mm}
%!!!!!!!!!!!!!!!!!!!!!!!!!!!!!!!!!!!!!!!!!

Field-programmable gate arrays (FPGAs) are the leading technology in
high-\linebreak throughput regex matching.
%
% >>>> TV: Nice but not enough space.
%
% They feature a~matrix of logic elements that can be configured to behave like
% an~arbitrary combinatorial or sequential digital circuit.
%
Due to their inherent parallelism, FPGAs provide an efficient way of
implementing \emph{nondeterministic finite automata} (NFAs), which naturally
arise from the input regexes.
%
% >>>> TV: Nice but not enough space.
%
% In particular, every state of an~NFA is mapped to a~flip-flop (FF) and
% transitions are encoded into combinatorial networks connecting these FFs.
% When reading input, all FFs that can be reached over the input are set to~1,
% performing, in fact, on-the-fly determinization (for the concrete word only)
% of the~NFA.  The input is accepted if at least one final-state FF is set to~1.
%
% >>>> TV: Nice but not enough space.
%
% The reconfigurability of FPGAs suits well the need to periodically update the
% database of the regular expressions as new network attacks appear). 
%
Although the amount of available resources in FPGAs is continually increasing,
%
% >>>> TV: Nice but not enough space.
%
% according to the Moore's law (which states that the number of transistors on
% a~chip doubles every 18~months), 
%
the speed of networks grows even faster.
%
% >>>> TV: Nice but not enough space.
%
% (Butters' law of photonics states that the speed of networks doubles every
% 9~months).
%
Working with multi-gigabit networks requires the hardware to use many parallel
packet processing branches in a single FPGA~\cite{Matousek2016}; each of them
implementing a separate copy of the concerned~NFA, and so reducing the size of
the NFAs is of the utmost importance. 
Various language-preserving automata reduction approaches exist, mainly based
on computing (bi)simulation relations on automata states (cf. the related work). 
The reductions they offer, however, do not satisfy the needs of high-speed
hardware-accelerated NIDSes.
% Their capabilities, however, do not satisfy the needs of high-speed
% hardware-accelerated NIDSes.

Our answer to the problem is \emph{approximate reduction} of NFAs, allowing for
a~trade-off between the achieved reduction and the precision of the regex
matching.
To formalise the intuitive notion of precision, we propose a novel
\emph{probabilistic distance} of automata.
It captures the probability that a packet of the input network traffic is
incorrectly accepted or rejected by the approximated NFA.
The distance assumes a~\emph{probabilistic model} of the network traffic (we
show later how such a model can be obtained).

Having formalised the notion of precision, we specify the target of our
reductions as two variants of an optimization problem: 
(1) minimizing the NFA size given the maximum allowed error (distance from the
original), or (2) minimizing the error given the maximum allowed NFA size.
Finding such optimal approximations is, however, computationally hard
(\pspace{}-complete, the same as precise NFA minimization).

Consequently, we sacrifice the optimality and, motivated by the typical
structure of NFAs that emerge from a set of regexes used by NIDSes (a~union of many long 
``tentacles''
%
% ``noodles'' 
%
with occasional small strongly-connected components), we limit the space of
possible reductions by restricting the set of operations they can apply to the
original automaton.
Namely, we consider two reduction operations: 
\begin{inparaenum}[(i)] \item collapsing the future of a state into a
\emph{self-loop} (this reduction over-approximates the language), or
\item \emph{removing states} (such a~reduction is under-approximating).
\end{inparaenum}
%
% TV: removed since it now seems not to fit nicely here.
%
% As the NFA in hardware is only used for pre-filtering to decrease the number
% of packets entering the software part of the NIDS, we primarily focus on the
% over-approximation.

The problem of identifying the optimal sets of states on which these operations
should be applied is still \pspace{}-complete.
The restricted problem is, however, more amenable to an approximation by
a~\emph{greedy algorithm}.
The algorithm applies the reductions state-by-state in an order determined by a
precomputed \emph{error labelling} of the states.
The process is stoppped once the given optimization goal in terms of the size or
error is reached.
The labelling is based on the probability of packets that may be accepted
through a given state and hence over-approximates the error that may be caused
by applying the reduction at a given state.
As our experiments show, this approach can give us high-quality reductions while
ensuring formal error bounds.

Finally, it turns out that even the pre-computation of the error labelling of
the states is costly (again \pspace{}-complete).
Therefore, we propose several ways to cheaply over-approximate it such that the
strong error bound guarantees are still preserved.
Particularly, we are able to exploit the typical structure of the ``union of
tentacles'' of the hardware NFA in an algorithm that is exponential in the size
of the largest ``tentacle'' only, which is indeed much faster in practice.

We have implemented our approach and evaluated it on regexes used to classify
malicious traffic in \snort{}. 
We obtain quite encouraging experimental results demonstrating that our approach
provides a much better reduction than language-preserving techniques with an
almost negligible error.
In particular, our experiments, going down to the level of an actual
implementation of NFAs in FPGAs, confirm that we can squeeze into an up-to-date
FPGA chip real-life regexes encoding malicious traffic, allowing them to be used
with a negligible error for filtering at speeds of 100\,Gbps (and even
400\,Gbps).
This is far beyond what one can achieve with current exact reduction approaches.

\enlargethispage{4mm}
%!!!!!!!!!!!!!!!!!!!!!!!!!!!!!!!!

%------------------------------------------------------------------------------
\vspace{-3mm}
\paragraph{Related Work}
Hardware acceleration for regex matching at the line rate is an
intensively studied technology that uses general-purpose hardware
\cite{DBLP:conf/sigcomm/KumarDYCT06,DBLP:conf/isca/TanS05,DBLP:conf/ancs/KumarTW06,DBLP:conf/conext/BecchiC07,DBLP:conf/ancs/BecchiC07,DBLP:conf/ancs/KumarCTV07,DBLP:conf/ancs/YuCDLK06,liu_fast,luchaup2014deep}
as well as
FPGAs~\cite{mitra_compiling,brodie_scalable,clark_efficient,hutchings_assign,sidhu_fast,Pus2011,Korenek2007,Kastil2009,Matousek2016}.
Most of the works focus on DFA implementation and optimization
techniques.
NFAs can be exponentially smaller than DFAs but need, in the worst case,
$\bigOof n$ memory accesses to process each byte of the payload where $n$~is the
number of states.
In most cases, this incurs an unacceptable slowdown.
Several works alleviate this disadvantage of NFAs by exploiting
reconfigurability and fine-grained parallelism of FPGAs,
allowing one to process one character per clock cycle (e.g.
\cite{mitra_compiling,brodie_scalable,sidhu_fast,Pus2011,Korenek2007,Kastil2009,Matousek2016}).

In \cite{luchaup2014deep}, which is probably the closest work to ours, 
the authors consider a set of regexes describing network attacks.
They replace a potentially prohibitively large DFA by a tree of smaller DFAs, 
an alternative to using NFAs
that minimizes the latency occurring in a non-FPGA-based implementation.
The language of every DFA-node in the tree over-approximates the languages of its children.
Packets are filtered through the tree from the root downwards until they belong
to the language of the encountered nodes, and may be finally accepted at the
leaves, or are rejected otherwise.
The over-approximating DFAs are constructed using a similar notion of
probability of an occurrence of a state as in our approach.
The main differences from our work are that 
(1) the approach targets approximation of DFAs (not NFAs),
(2) the over-approximation is based on a given traffic sample only (it cannot
benefit from a probabilistic model), and 
(3) no probabilistic guarantees on the approximation error are provided.

Approximation of DFAs was considered in various other contexts.
Hyper-mi\-ni\-mi\-za\-tion is an approach that is allowed to alter language
membership of a finite set of words~\cite{maletti_hypermini,jez_hypermini}.
A DFA  with a given maximum number of states is constructed in~\cite{ganty_budget},
minimizing the error defined either by (i) counting prefixes of misjudged words up to some length, 
or (ii) the sum of the probabilities of the misjudged words wrt the Poisson
distribution over $\Sigma^*$.
Neither of these approaches considers reduction of NFAs nor 
allows to control the expected error with respect to the real traffic.
%takes into
%consideration the probability of individual words in the real traffic.
%consideration the statistical significance of individual words in the real traffic.
%
%\tv{Point (ii) speaks about probability: it is not probability of words???}
%\lh{changed, better?}

In addition to the metrics mentioned above when discussing the works
\cite{ganty_budget,maletti_hypermini,jez_hypermini}, the following metrics
should also be mentioned.
The Cesaro-Jaccard distance studied in~\cite{parker_cj_distance} is, in
spirit, similar to \cite{ganty_budget} and does also not reflect the probability
of individual words.
The edit distance of weighted automata from \cite{mohri_edit_distance} depends
on the minimum edit distance between pairs of words from the two compared
languages, again regardless of their statistical significance.
None of these notions is suitable for our needs.

%% Every leaf node $D_i$ accepts the union of a subset of the regular expressions.
%% Every $D_i$ is then reduced to $D'_i$ such that $D'_i$ over-approximates $D_i$
%% wrt a given error.
%% A node in the second layer is a~DFA that accepts the union of a subset of DFAs $D'_i$.
%% The other layers are built in the same way, alternating between
%% %
%% (i) reducing DFAs on the layer directly below and
%% (ii) uniting the results of the reduction,
%% %
%% until only a single DFA is obtained.
%% (Note that only results of steps (ii) are kept in the resulting DFA-tree.)
%% The matching process is performed in the top-down fashion.
%% If a DFA $D$ in the tree rejects a packet, it is guaranteed that the packet does not match
%% any regular expression in the leaves of the subtree rooted in~$D$.
%% If $D$ accepts a packet, the packet is sent to children nodes to be matched using more precise DFAs.
%% The DFA-tree accepts the packet only if a leaf node accepts it.
%% Therefore, the DFA-tree accepts exactly the union of the set of regular expressions.
%% In contrast to our work, this approach does not aim at memory reduction, it
%% actually introduces a memory overhead.
%% 
%% Further, the particular over-approximations are based only on the given traffic
%% and no probabilistic guarantees on the approximation error are provided.
%% As a~result, our work provides a more general methodology, since we learn
%% a~model from the traffic and provide bounds on the approximation error with
%% respect to the model.

Language-preserving minimization of NFAs is a~\pspace{}-complete problem \cite{Jiang1993,malcher_minimizing_2004}.
More feasible (polynomial-time) is language-preserving
size reduction of NFAs based on (bi)simulations
\cite{hopcroft_nlogn_1971,paige_three_1987,bustan_simulation_2003,Champarnaud2004},
which does not aim for a truly minimal NFA.
%and is by merging states and removing transitions based on compatibility wrt (bi)simulation relations 
A~number of advanced variants exist, based on multi-pebble or look-ahead
simulations, or on combinations of forward and backward
simulations~\cite{mayr_advanced_2013,etessami_hierarchy,lorenzo_advanced}.
%Some of them are implemented in the tool REDUCE  \cite{REDUCE}.
%
The practical efficiency of these techniques is, however, often insufficient to
allow them to handle the large NFAs that occur in practice and/or they do not
manage to reduce the NFAs enough.
Finally, even a~minimal NFA for the given set of regexes
is often too big to be implemented in the given FPGA operating on the required
speed (as shown even in our experiments). 
Our approach is capable of a~much better reduction for the price of a~small
change of the accepted language.

%The technique computing the probabilistic distance between two UFAs with respect to a probabilistic automaton PA, which represents an important part of the state labelling process, builds on the method that computes the probability of the language of an UFA given by the PA. This method is closely related to a more general approach for model checking of Markov chains against unambiguous B\"{u}chi automata published in~\cite{Baier2016} 

%%%%%%%%%%%%%%%%%%%%%%%%%%%%%%%%%%%%%%%%%%%%%%%%%%%%%%%%%%%%%%%%%%%%%%%%%%%%%%%%
\vspace{-3.0mm}
\section{Preliminaries}\label{sec:label}
\vspace{-2.0mm}
%%%%%%%%%%%%%%%%%%%%%%%%%%%%%%%%%%%%%%%%%%%%%%%%%%%%%%%%%%%%%%%%%%%%%%%%%%%%%%%%

We use $\range a b$ to denote the set $\{x \in \real \mid a \leq x \leq b\}$ and $\nat$ to denote the set $\{ 0,1,2,\dots \}$.
% and $\rangeN a b$ to denote $\range a b \cap \nat$ \td{OL: needed?}.
Given a pair of sets $X_1$ and $X_2$, we use $X_1 \symdiff X_2$ to denote their
\emph{symmetric difference}, i.e., the set $\{x \mid \exists! i \in \{1,2\} : x
\in X_i\}$.
We use the notation $\vecof{v_1, \ldots, v_n}$ to denote a~vector of $n$~elements,
$\onevec$ to denote the all 1's vector~$\vecof{1, \ldots, 1}$,
$\bA$ to denote a~matrix, and $\bA^\top$ for its transpose, and $\bI$ for the
identity matrix.

%\rem{
%A~\emph{nondeterministic finite automaton} (NFA) over a finite
%alphabet~$\Sigma$ is a~quintuple $A = (Q, \Sigma, \trans, q_0, F)$ where $Q$ is
%a~finite non-empty set of states, $\trans: Q\times\Sigma\rightarrow 2^Q$ is
%a~transition function, $q_0 \in Q$ is an initial state, and $F\subseteq Q$ is
%a~set of accepting states.
%A~\emph{configuration} of $A$ is an element from $Q\times \Sigma^*$.
%A~\emph{step} of $A$ is a~binary relation $\vdash_A$ on the set of
%configurations  defined as $(q,ax)\vdash_A(p,x)$ iff $p\in\trans(q,a)$.
%$A$~\emph{accepts} the  language $L(A)$  defined as $L(A) = \{ w\in\Sigma^*\ |\
%(q_0, w) \vdash_A^* (p, \varepsilon),p\in F \}$.
%The \emph{language accepted at $q \in Q$}  is defined as $ L^{-1}_A(q) = \{
%w\in\Sigma^*\ |\ (q_0, w) \vdash_A^* (q, \varepsilon) \}$.
%A finite automaton $A $ is called \emph{deterministic} (DFA) if $\forall q\in
%Q$ and $\forall a\in\Sigma:$ $|\trans(q,a)|\leq 1$.
%An NFA is called \emph{unambiguous} (UFA) if it has at most one accepting
%computation on every input string, where an \emph{accepting computation on
%$x$} is a~finite sequence $C_1, \dots, C_k$ of configurations such that $C_i
%\vdash_A C_{i+1}$ for all $i = 1,\dots,k-1$, $C_1 = (q_0,x)$, and $C_m = (p,
%\varepsilon)$ for a~state $p \in F$.
%}

In the following, we fix a~finite non-empty alphabet~$\Sigma$.
A~\emph{nondeterministic finite automaton} (NFA) is a~quadruple
$A = (Q, \trans, I, F)$ where $Q$ is
a~finite set of states, $\trans: Q\times\Sigma \to 2^Q$ is
a~transition function, $I \subseteq Q$ is a~set of initial states, and
$F\subseteq Q$ is a~set of accepting states.
% \tv{(1) Either quadruple OR include $\Sigma$. (2) May be quadruple, but then say
% at beginning that we assume a fixed alphabet $\Sigma$. (3) Then check the
% appearances to have it uniform.}
We use $\statesof A, \transof A, \initof A$, and~$\finof A$ to denote $Q,
\trans, I$, and $F$, respectively, and $q \ltr a q'$ to denote that $q' \in
\trans(q, a)$.
A~sequence of states $\rho = q_0 \cdots q_n$ is a~\emph{run} of $A$ over
a~word~$w = a_1 \cdots a_n \in \Sigma^*$ from a~state~$q$ to a~state $q'$,
denoted as $q \runstoover{w, \rho} q'$, if $\forall 1 \leq i \leq n: q_{i-1}
\ltr{a_i} q_i$, $q_0 = q$, and $q_n = q'$.
Sometimes, we use~$\rho$ in set operations where it behaves as the set of
states it contains.
We also use $q \runstoover w q'$ to denote that $\exists \rho \in Q^*: q
\runstoover{w,\rho} q'$
and $q \runsto q'$ to denote that $\exists w: q \runstoover w q'$.
The \emph{language} of a~state~$q$ is defined as
$\langinof A q = \{w \mid \exists q_F \in F: q \runstoover w q_F\}$
and its \emph{banguage} (back-language) is defined as
$\blanginof A q = \{w \mid \exists q_I \in I: q_I \runstoover w q\}$.
Both notions can be naturally extended to a~set $S \subseteq Q$: $\langinof A S
= \bigcup_{q\in S} \langinof A q$ and $\blanginof A S
= \bigcup_{q\in S} \blanginof A q$.
We drop the subscript $A$ when the context is obvious.
$A$~\emph{accepts} the language $\langof A$  defined as $\langof A = \langinof A
I$.
$A $ is called \emph{deterministic} (DFA) if $|I| = 1$ and $\forall q\in Q$ and
$\forall a\in\Sigma:$ $|\trans(q,a)|\leq 1$,
and \emph{unambiguous} (UFA) if $\forall w \in \langof A: \exists! q_I \in I,
\rho \in Q^*, q_F \in F: q_I \runstoover{w, \rho} q_F$.

The \emph{restriction} of~$A$ to $S \subseteq Q$ is an~NFA~$\restr A S$
given as~$\restr A S = (S, \trans \cap (S \times \Sigma \times 2^S),
I \cap S, F \cap S)$.
We define the \emph{trim} operation as $\trimof{A} = \restr{A}{C}$ 
where $C = \{ q \mid \exists q_I\in I, q_F\in F: q_I \runsto q \runsto q_F \}$.
For a~set of states $R \subseteq Q$, we use $\reachof R$ to denote the set of
states reachable from~$R$, formally, $\reachof R = \{r' \mid \exists r \in R: r
\runsto r'\}$.
% , and $\coneof R$ to denote the subset of~$\reachof R$ reachable
% \emph{only} from~$R$, i.e., $\coneof R = \reachof R \setminus \reachof{Q
% \setminus R}$.
We use the number of states as the measurement of the size of~$A$, i.e., $|A|
= |Q|$.

A (discrete probability) \emph{distribution} over a set $X$ is a~mapping
$\distr: X \to \range 0 1$ such that $\sum_{x \in X} \distr(x) = 1$.
% \tv{The sum works for discrete distributions only?}
% A (probability) \emph{semi-distribution} \td{OL: needed?} over a set $X$ is a
% mapping
% $\distr: A \to \range 0 1$ such that $\sum_{x \in X} \distr(x) \leq 1$.
% If $\sum_{x \in X} \distr(x) = 1$, we call $\distr$ a~\emph{distribution}.
An~$n$-state \emph{probabilistic automaton} (PA) over $\Sigma$ is a triple
$P = (\pinit, \pfin, \{\ptransa\}_{a \in \Sigma})$
where
$\pinit \in \range{0}{1}^n$ is a vector of \emph{initial weights}, 
$\pfin \in \range{0}{1}^n$ is a vector of \emph{final weights}, and
for every $a \in \Sigma$,
$\ptransa \in \range{0}{1}^{n\times n}$ is a \emph{transition matrix} for
symbol~$a$.
% \tv{Is transition matrix defined somewhere?}
% \td{OL: just been defined :-)}
We abuse notation and use $\statesof P$ to denote the set of states
$\statesof P = \{1, \ldots, n\}$.
Moreover, the following two properties need to hold:
\begin{inparaenum}[(i)]
  \item  $\sum \{\pinitof i \mid i \in \statesof P\} = 1$ (the initial
    probability is~1) and
  \item  for every state $i \in \statesof P$ it holds that
    $\sum\{\ptransaof i j \mid j\in \statesof P, a\in \Sigma\} + \pfinof i = 1$
    (the probability of accepting or leaving a state is~1).
\end{inparaenum}
We define the \emph{support} of $P$ as the NFA $\suppof P = (\statesof P,
\transof P, \initof P, \finof P)$~s.t.
%
% \begin{align}
%   \trans_P  &= \{(i, a, j) \mid \ptransaof i j > 0\} &
%   \initof P &= \{i \in \statesof{P} \mid \pinit[i] > 0\} &
%   \finof P  &= \{i \in \statesof{P} \mid \pfin[i] > 0\}
% \end{align}
\vspace{-1mm}
\begin{align*}
  \transof P  &= \{(i, a, j) \mid \ptransaof i j > 0\} &
  \initof P   &= \{i \mid \pinit[i] > 0\} &
  \finof P    &= \{i \mid \pfin[i] > 0\}.\\[-6mm]
\end{align*}
Let us assume that every PA~$P$ is such that~$\suppof P = \trimof{\suppof P}$.
For a~word~$w = a_1 \ldots a_k \in \Sigma^*$, we use $\ptrans_w$ to denote the
matrix~$\ptrans_{a_1} \cdots \ptrans_{a_k}$.
It can be easily shown
% \tv{Well... Is it in fact guaranteed that it will really
% generate a distribution? What guarantees that the sum will be one?}
that~$P$ represents a distribution over words~$w \in
\Sigma^*$ defined as
$\distrP(w) = \pinit^\top \cdot \ptrans_w \cdot \pfin$.
We call~$\distrP(w)$ the \emph{probability} of~$w$ in~$P$.
Given a~language $L \subseteq \Sigma^*$, we define the probability of~$L$
in~$P$ as~$\distrP(L) = \sum_{w \in L} \distrP(w)$.
%
% TV: The below seems to be used in one place only, so defined there (together
% with the intuition).
%
% We use $\ptrans$ to denote the sum $\sum_{a\in\Sigma} \ptransa$.

If Conditions (i) and (ii) from the definition of PAs are dropped, we speak
about a~\emph{pseudo-probabilistic automaton (PPA)}, which may assign a word from
its support a~quantity that is not necessarily in the range $\range{0}{1}$,
denoted as the~\emph{significance} of the word below.
PPAs may arise during some of our operations performed on PAs.

%!!!!!!!!!!!!!!!!!!!!!!!!!!!!!!!!!!
\enlargethispage{3mm}
%!!!!!!!!!!!!!!!!!!!!!!!!!!!!!!!!!!

%%%%%%%%%%%%%%%%%%%%%%%%%%%%%%%%%%%%%%%%%%%%%%%%%%%%%%%%%%%%%%%%%%%%%%%%%%%%%%%%
\vspace{-3.0mm}
\section{Approximate Reduction of NFAs}\label{sec:label}
\vspace{-2.0mm}
%%%%%%%%%%%%%%%%%%%%%%%%%%%%%%%%%%%%%%%%%%%%%%%%%%%%%%%%%%%%%%%%%%%%%%%%%%%%%%%%

In this section, we first introduce the key notion of our approach:
a~\emph{probabilistic distance} of a~pair of finite automata wrt a~given
probabilistic automaton that, intuitively, represents the significance of
particular words. We discuss the complexity of computing the probabilistic
distance. Finally, we formulate two problems of \emph{approximate automata
reduction via probabilistic distance}.

%*******************************************************************************
\vspace{-3.0mm}
\subsection{Probabilistic Distance}\label{sec:prob-distance}
\vspace{-1.5mm}
%*******************************************************************************

We start by defining our notion of a probabilistic distance of two NFAs.
Assume NFAs $A_1$ and $A_2$
%over the alphabet $\Sigma$
and
a~probabilistic automaton~$P$ specifying the distribution~$\distrP: \Sigma^* \to
\range 0 1$.
The \emph{probabilistic distance} $\distanP(A_1,A_2)$ between $A_1$ and $A_2$
wrt $\distrP$ is defined~as
%
%\vspace{-3mm}
\begin{equation*}
  \distanP(A_1,A_2) = \distrP(\langof{A_1} \symdiff
  \langof{A_2}).%\\[-1mm]
\end{equation*}
Intuitively, the distance captures the significance of the words accepted by one
of the automata only.
We use the distance to drive the reduction process towards automata with small
errors and to assess the quality of the resulting automata.

The value of $\distrP(\langof{A_1} \symdiff \langof{A_2})$ can be computed as
follows.
Using the fact that (1) $L_1 \symdiff L_2 =  (L_1 \setminus L_2) \dunion (L_2
\setminus L_1)$ and (2) $L_1 \setminus L_2 = L_1 \setminus (L_1 \cap L_2)$, we
get%\vspace*{-2mm}
\begin{equation*}
  \begin{split}
  \distanP(A_1,A_2) &=
    \distrP(\langof{A_1} \setminus \langof{A_2}) +
    \distrP(\langof{A_2} \setminus \langof{A_1}) \\
  &=
    \distrP(\langof{A_1} \setminus (\langof{A_1}\cap \langof{A_2})) +
    \distrP(\langof{A_2} \setminus (\langof{A_2}\cap \langof{A_1})) \\
    &=
      \distrP(\langof{A_1}) +
      \distrP(\langof{A_2}) -
      2\cdot\distrP(L(A_1) \cap L(A_2)).%\\[-3mm]
  \end{split}
\end{equation*}
Hence, the key step is to compute $\distrP(L(A))$ for an NFA~$A$ and a PA
$P$. Problems similar to computing such a probability have been extensively
studied in several contexts including verification of probabilistic
systems~\cite{Vardi1985,Baier2016,Baier2016A}.
The below lemma summarises the complexity of this step.

%!!!!!!!!!!!!!!!!!!!!!!!!!!!!!!!!!!!!!!
%\enlargethispage{5mm}
%!!!!!!!!!!!!!!!!!!!!!!!!!!!!!!!!!!!!!!

\begin{restatable}{lemma}{polytimeLangProb}\label{lem:polytime-lang-prob}
Let $P$ be a PA and $A$ an NFA.
The problem of computing $\distrP(\langof{A})$ is \pspace{}-complete.
For a~UFA $A$, $\distrP(\langof{A})$ can be computed~in~\ptime{}.
\end{restatable}

In our approach, we apply the method of~\cite{Baier2016A} and
compute $\distrP(L(A))$ in the following way. We first check whether the NFA $A$ is
unambiguous.
This can be done by using the standard product construction (denoted as $\cap$)
for computing the intersection of the NFA $A$ with itself and trimming the result,
formally $B = \trim(A \cap A)$,
% constructing an NFA $B = \trim(A \cap A)$,
followed by
a~check whether there is some state $(p,q)\in Q[B]$ s.t. $p \neq q$
\cite{Mohri2012}.
If $A$ is ambiguous, we either determinise it or disambiguate it
\cite{Mohri2012}, leading to a DFA/UFA $A'$, respectively.%
\footnote{In theory, disambiguation can produce smaller automata, but, in our
experiments, determinisation proved to work better.}
Then, we construct the trimmed product of $A'$ and $P$ (this can be seen as
computing $A' \cap \suppof P$ while keeping the probabilities from~$P$ on the
edges of the result), yielding a~PPA $R =
(\pinit, \pfin, \{\ptransa\}_{a \in \Sigma})$.\footnote{$R$ is not necessarily a
PA since there might be transitions in $P$ that are either removed or copied
several times in the product construction.}
Intuitively, $R$~represents not only the words of $L(A)$ but also their
probability in $P$.
Now, let $\ptrans = \sum_{a \in \Sigma} \ptransa$ be the matrix that expresses,
for any $p,q \in Q[R]$, the significance of getting from $p$ to $q$ via any $a
\in \Sigma$. 
Further, it can be shown (cf.~the proof of Lemma~\ref{lem:polytime-lang-prob}
in the Appendix)
that the matrix $\ptrans^*$, representing the significance of going from $p$ to
$q$ via any $w \in \Sigma^*$, can be computed as $(\bI - \ptrans)^{-1}$.
Then, to get $\distrP(L(A))$, it suffices to take $\pinit^\top\cdot
\ptrans^*\cdot\pfin$.
Note that, due to the determinisation/disambiguation step, the obtained value
indeed is $\distrP(L(A))$ \mbox{despite~$R$ being a~PPA.}

%*******************************************************************************
\vspace{-0.0mm}
\subsection{Automata Reduction using Probabilistic Distance}\label{sec:label}
\vspace{-0.0mm}
%*******************************************************************************

We now exploit the above introduced probabilistic distance to formulate the task
of approximate reduction of NFAs as the following two optimisation problems.
Given an~NFA~$A$ and a~PA~$P$ specifying the distribution $\distrP: \Sigma^* \to
\range 0 1$, we define%\vspace{-2mm}
\begin{itemize}

  \item  \textbf{size-driven reduction}: for $n \in \nat$, find an~NFA~$A'$ such
  that~$|A'| \leq n$ and the distance $\distanP(A, A')$ is minimal,
  
  \item  \textbf{error-driven reduction}: for $\epsilon \in \range 0 1$, find
  an~NFA~$A'$ such that $\distanP(A, A') \leq \epsilon$ and the size~$|A'|$
  is minimal.%\vspace*{-2mm}

\end{itemize}
The following lemma shows that the natural decision problem underlying both of
the above optimization problems is \pspace{}-complete, which matches the
complexity of computing the probabilistic distance as well as that of the
\emph{exact} reduction of NFAs~\cite{Jiang1993}.%\vspace*{-1mm}

\begin{restatable}{lemma}{generalPspace}\label{lem:general-pspace}
  Consider an NFA~$A$, a~PA~$P$, a~bound
on the number of states $n \in \nat$, and an~error bound~$\epsilon \in \range 0
1$. 
It is \pspace{}-complete to determine whether there exists an~NFA~$A'$ with~$n$
states s.t.~$\distanP(A, A') \leq \epsilon$.
\end{restatable}

%\vspace*{-1mm}
The notions defined above do not distinguish between introducing a \emph{false
positive} ($A'$~accepts a~word $w \notin \langof A$) or a \emph{false negative}
($A'$~does not accept a~word $w \in \langof A$) answers.
To this end, we define \emph{over-approximating} and \emph{under-approximating}
reductions as reductions for which the additional conditions $\langof A
\subseteq \langof{A'}$ and $\langof A \supseteq \langof{A'}$ hold, respectively.

%A~na\"{i}ve solution to the reductions that enumerates all NFAs $A'$ of
%sizes from 0 up to $k$ (resp.~$|A|$) and takes a smallest one with the minimal $\distanP(A,
%A')$ (resp.~a smallest one satisfying the restriction on $\distanP(A, A')$)
%is obviously infeasible. 
%
%We will propose practical heuristics in the next section.

 A~na\"{i}ve solution to the reductions would enumerate all NFAs $A'$ of
 sizes from 0 up to $k$ (resp.~$|A|$), for each of them compute $\distanP(A,
 A')$, and take an automaton with the smallest probabilistic distance
 (resp.~a smallest one satisfying the restriction on $\distanP(A, A')$).
 Obviously, this approach is computationally infeasible.

% A~na\"{i}ve solution to a size-driven reduction would enumerate all NFAs $A'$ of
% sizes $0, 1, \ldots, n$, for each of them compute $\distanP(A, A')$, and take
% the automaton with the smallest probabilistic distance.
% %
% For an error-driven reduction, on the other hand, a~na\"{i}ve solution would
% enumerate all NFAs of the sizes~$0,1,\ldots,|A|$ and pick the smallest one
% satisfying the restriction on the probabilistic distance.
% %
% Obviously, neither of these approaches is computationally feasible.

%!!!!!!!!!!!!!!!!!!!!!!!!!!!!!!!!
% \enlargethispage{4mm}
%!!!!!!!!!!!!!!!!!!!!!!!!!!!!!!!!

%%%%%%%%%%%%%%%%%%%%%%%%%%%%%%%%%%%%%%%%%%%%%%%%%%%%%%%%%%%%%%%%%%%%%%%%%%%%%%%%
\vspace{-0.0mm}
%\section{Approximate Reductions using State Labelling\\[1em]
\section{A Heuristic Approach to Approximate Reduction}\label{sec:reductions}
\vspace{-0.0mm}
In this section, we introduce two techniques for approximate reduction of NFAs
that avoid the need to iterate over all automata of a~certain size.
The first approach under-approximates the automata by removing states---we
call it the \emph{pruning reduction}---while
the second approach over-approximates the automata by adding self-loops to
states and removing redundant states---we call it the
\emph{self-loop reduction}.
Finding an optimal automaton using these reductions is also
\pspace{}-complete, but more amenable to heuristics like greedy algorithms.
We start with introducing two high-level greedy algorithms, one for the size-
and one for the error-driven reduction, and follow by showing their
instantiations for the pruning and the self-loop reduction.
%
% A crucial role in all these algorithms is played by a function that labels states
A~crucial role in the algorithms is played by a function that labels states
of the automata by an estimate of the error that will be caused when some
of the reductions is applied at a~given~state.

% The approaches give results that are sub-optimal, but still provide the
% required error guarantees and allow the use of an efficient greedy algorithm.
% % driven by probabilistic distance between the automata wrt the
% % probabilistic model given by a~DPA~$P$.
% Both reductions are based on computing a~state labelling that captures how the
% modification of the particular states increases the worst-case error.
% At the end of the section, we show how to compute such a~labelling efficiently.
%
% \td{OL: say that we start by introducing reduction-agnostic meta-algorithm first}

%*******************************************************************************
\vspace{-3.0mm}
\subsection{A General Algorithm for Size-Driven Reduction}
\label{sec:greedy-size}
\vspace{-1.5mm}
%*******************************************************************************

\begin{wrapfigure}[10]{r}{6.8cm}
\SetAlCapHSkip{0mm}
\vspace{-8.5mm}
\hspace{-2mm}
\begin{minipage}{6.95cm}
\setlength{\algomargin}{0mm}
\SetAlCapHSkip{0mm}
\IncMargin{3mm}
\begin{algorithm}[H]
  \SetKwInOut{Input}{Input}\SetKwInOut{Output}{Output}
  \Input{\mbox{NFA $A = (Q, \trans, I, F)$, PA $P$, $n \geq 1$}}
  \Output{NFA $A'$, $\epsilon \in \real$ s.t. $|A|
          \leq n$ and $\distanP(A, {A'}) \leq \epsilon$\vspace*{-1mm}}
  % \BlankLine
  $V \gets \emptyset$\;
  \For{$q \in Q $ in the order $\ordalglab$} {
    $V \gets V \cup \{q\}$;
    $A' \gets \reduceof{A}{V}$\;
    \lIf{$|A'| \leq n$} {
      \textbf{break}
    }
  }
  \mbox{\Return $A'$, $\epsilon = \errfuncof A V
  {\labfuncof A P}$\;}\\[0mm]
\caption{\mbox{A greedy size-driven reduction}}
\label{alg:general-greedy-size-reduction}
\end{algorithm}
\end{minipage}
\end{wrapfigure}

Algorithm~\ref{alg:general-greedy-size-reduction} shows a~general greedy method
for performing the size-driven reduction.
In order to use the same high-level algorithm in both directions of reduction
(over/under-approximating), it is parameterized with three functions: $\labfunc,
\reduce$, and $\errfunc$.
The
real intricacy of the procedure is hidden inside these three functions.
Intuitively, $\labfuncof A P$ assigns every state of an NFA~$A$ an
approximation of the error that will be caused wrt the PA~$P$ when
a reduction is applied at this state,
while the purpose 
of $\reduceof A V$ is to create a~new NFA~$A'$ obtained from~$A$ by
introducing some error at states from~$V$.%
\footnote{We emphasize that this does not mean that states from $V$ will be simply
removed from~$A$---the performed operation depends on the particular reduction.}
Further, $\errfuncof A V {\labfuncof A P}$
estimates the error introduced by the application of $\reduceof A V$,
possibly in a more precise (and costly) way than by just summing the concerned
error labels: Such a computation is possible outside of the main computation
loop.
We show instantiations of these functions later, when discussing the
reductions used.
Moreover, the algorithm is also parameterized with a~total order $\ordalglab$ that
defines which states of $A$ are processed first and which are processed later.
The ordering may take into account the precomputed labelling.
The algorithm accepts an NFA~$A$, a~PA~$P$, and $n \in \nat$ and outputs a~pair
consisting of an~NFA~$A'$ of the size $|A'| \leq n$ and an error bound $\epsilon$
such that $\distanP(A, A') \leq \epsilon$.

The main idea of the algorithm is that it creates a~set~$V$ of states where an
error is to be introduced.
$V$ is constructed by starting from an empty set and adding states to it in the
order given by~$\ordalglab$, until the size of the result of $\reduceof A V$ has
reached the desired bound~$n$
(in our setting, $\reduce$ is always antitone, i.e., for $V \subseteq
V'$, it holds that $|\reduceof A V| \geq |\reduceof A {V'}|$).
% first computes an approximation of
% the significance of states of~$A$ as a~\emph{state labelling} $\stlab:
% \statesof A \to \real$.
% It proceeds by computing the set of states $V$ where an error will be introduced,
% starting from the set of all states~$Q$ and removing states in a~descending
% order of significance as long as the size of $\reduceof A V$ is below the
% threshold~$k$.
%
% it loops over all states $q$ of $A$, starting with the most significant
% ones, and adds~$q$ into a~set~$V$ if this addition does not cause the result of
% $\reduceof A {Q \setminus (V \cup \{q\})}$---which should create an NFA obtained
% from $A$ by removing $Q \setminus (V \cup \{q\})$---be larger than~$k$.
We now define the necessary condition for $\labfunc, \reduce$, and $\errfunc$
that makes Algorithm~\ref{alg:general-greedy-size-reduction} correct.
%

% \smallskip
% \noindent
% \textbf{Condition \cone{}}
% {
%   \it
% For every NFA~$A$, PA~$P$, and a set~$V \subseteq \statesof A$, it holds that
% $\errfuncAVL \geq \distanP(A, \reduceof{A}{V})$, where $\stlab = \labfuncof A
% P$.
% }
\medskip
\begin{condition}{\cone{}}
holds if for every NFA~$A$, PA~$P$, and a set~$V \subseteq \statesof A$, we have that
\begin{inparaenum}[(a)]
  \item  $\errfuncof A V {\labfuncof A P} \geq \distanP(A, \reduceof{A}{V})$,
  \item  $|\reduceof A {Q[A]}| \leq 1$, and
  \item  $\reduceof A {\emptyset} = A$.
\end{inparaenum}
% For every NFA~$A$, PA~$P$, and a set~$V \subseteq \statesof A$, it holds that
% $\errfuncAVL \geq \distanP(A, \reduceof{A}{V})$, where $\stlab = \labfuncof A
% P$, $|\reduceof A {Q[A]}| \leq 1$, and $\reduceof A {\emptyset} = A$.
\end{condition}
\medskip

\cone{}(a) ensures that the error computed by the reduction algorithm
indeed over-approximates the exact probabilistic distance,
\cone{}(b) ensures that the algorithm can (in the worst case, by applying the
reduction at every state of~$A$) for any $n \geq 1$ output a~result~$|A'|$ of the
size $|A'|\leq n$, and
\cone{}(c) ensures that when no error is to be introduced at any state, we
obtain the original automaton.%\vspace*{-1mm}

\begin{lemma}\label{lem:correctness-greedy}
Algorithm~\ref{alg:general-greedy-size-reduction} is correct if \cone{} holds.
  % Algorithm~\ref{alg:general-greedy-size-reduction} is correct if for every
  % NFA~$A$, PA~$P$, and a set~$V \subseteq \statesof A$, it holds that
  % $\errfuncAVL \geq \distanP(A, \reduceof{A}{V})$, where $\stlab = \labfuncof A
  % P$.
\end{lemma}

\begin{proof}
Follows straightforwardly from Condition \cone{}.
\qed
\end{proof}

% \begin{proof}
% \td{TODO}
% \qed
% \end{proof}

% We show that \cone{} and \ctwo{} hold at every instantiation of the general
% algorithm later.

%!!!!!!!!!!!!!!!!!!!!!!!!!!!!!!!!!!!!!!
% \enlargethispage{6mm}
%!!!!!!!!!!!!!!!!!!!!!!!!!!!!!!!!!!!!!!

\pagebreak

%*******************************************************************************
\vspace{-0.0mm}
\subsection{A General Algorithm for Error-Driven Reduction}
\label{sec:greedy-error}
\vspace{-0mm}
%*******************************************************************************

\begin{wrapfigure}[10]{r}{7.1cm}
\begin{minipage}{7.2cm}
\setlength{\algomargin}{0mm}
\SetAlCapHSkip{0mm}
\vspace{-9mm}
\hspace{-2mm}
\IncMargin{3mm}
\begin{algorithm}[H]
  \SetKwInOut{Input}{Input}\SetKwInOut{Output}{Output}
  \Input{\mbox{NFA $A = (Q, \trans, I, F)$, PA $P$, $\epsilon \in \range 0 1$}}
  \Output{NFA $A'$ s.t. $d_P(A, A') \leq \epsilon$}
  % \BlankLine
  $\stlab \gets \labfuncof A P$\;
  $V \gets \emptyset$\;
  \For{$q \in Q $ in the order $\ordalglab$} {
    $e \gets \errfuncof{A}{V\cup \{q\}}{\stlab}$\;
    \lIf{$e \leq \epsilon$} {
      $V \gets V \cup \{q\}$
    }
  }
  \Return $A' = \reduceof{A}{V}$\;
\caption{\mbox{A greedy error-driven reduction.}}
\label{alg:general-greedy-error-reduction}
\end{algorithm}
\end{minipage}
\end{wrapfigure}
In Algorithm~\ref{alg:general-greedy-error-reduction}, we provide a~high-level
method of computing the error-driven reduction.
The algorithm is in many ways similar to
Algorithm~\ref{alg:general-greedy-size-reduction};
it also computes a~set of\linebreak
states~$V$ where an error is to be introduced, but an
important difference is that we compute an approximation of the error in each
step and only add $q$ to $V$ if it does not raise the error over the
threshold~$\epsilon$.
Note that the $\errfunc$ does not need to be monotone, so it may be advantageous
to traverse all states from~$Q$ and not terminate as soon as the threshold is
reached.
The correctness of Algorithm~\ref{alg:general-greedy-error-reduction} also
depends on \cone{}.\vspace*{-1mm}

\begin{lemma}\label{lem:err-driven-correct}
Algorithm~\ref{alg:general-greedy-error-reduction} is correct if \cone{} holds.
\end{lemma}

\begin{proof}
Follows straightforwardly from Condition \cone{}.
\qed
\end{proof}

\enlargethispage{6mm}
%!!!!!!!!!!!!!!!!!!!!!!!!!!!!!!!!!!!!!!

%*******************************************************************************
\vspace{-0.0mm}
\subsection{Pruning Reduction}\label{sec:pruning}
\vspace{-0mm}
%*******************************************************************************

The pruning reduction is based on identifying a~set of states to be removed
from an NFA~$A$, under-approximating the language of~$A$.
In particular, for~$A = (Q, \trans, I, F)$, the pruning reduction finds a~set
$R \subseteq Q$ and restricts~$A$ to~$Q\setminus R$, followed by removing
useless states, to construct a reduced automaton~$A' =
\trimof{\restr{A}{Q\setminus R}}$.
Note that the natural decision problem corresponding to this reduction is also
\pspace{}-complete.%\vspace*{-1mm}
\begin{restatable}{lemma}{pruningPspace}\label{lem:pruning-pspace}
  Consider an NFA~$A$, a~PA~$P$, a~bound on the number of states $n \in \nat$,
  and an~error bound~$\epsilon \in \range 0 1$.
  It is \pspace{}-complete to determine whether there exists a~subset of states
  $R \subseteq \statesof A$ of the size $|R| = n$ such that $\distanP(A,
  \restr A R) \leq \epsilon$.
\end{restatable}

%\vspace*{-1mm}
Although Lemma~\ref{lem:pruning-pspace} shows that the pruning
reduction is as hard as a~general reduction
(cf.~Lemma~\ref{lem:general-pspace}), the pruning reduction is more amenable to the
use of heuristics like the greedy algorithms from
\S\ref{sec:greedy-size} and \S\ref{sec:greedy-error}.
We instantiate $\reduce$, $\errfunc$, and $\labfunc$ in these
high-level algorithms in the following way (the subscript $p$ means
\emph{pruning}):
%
% \vspace{-2mm}
\begin{align*}
  \reducePof A V = \trimof{\restr A {Q\setminus V}},
  &&
  \errfuncPof A V \stlab = \min_{V' \in \downclosPof V}
    \sum \left\{\stlabof q \mid q\in V'\right\},%\\[-9mm]
\end{align*}
where $\downclosPof V$ is defined as follows.
Because of the use of $\trim$ in $\reduceP$, for a~pair of sets~$V, V'$ s.t.~$V
\subset V'$, it holds that $\reducePof A V$ may, in general, yield the same
automaton as $\reducePof A {V'}$.
Hence, we define a~partial order $\ordP$ on $2^Q$ as $V_1 \ordP V_2$ iff $\reducePof A
{V_1} = \reducePof A {V_2}$ and $V_1 \subseteq V_2$, and use $\downclosPof V$ to
denote the set of minimal elements wrt~$V$ and~$\ordP$.
The value of the approximation $\errfuncPof A V \stlab$ is therefore the minimum
of the sum of errors over all sets from $\downclosPof V$.

Note that the size of $\downclosPof V$ can again be exponential, and thus we
employ a greedy approach for guessing an optimal~$V'$. Clearly, this cannot 
affect the soundness of the algorithm, but only decreases the precision of the 
bound on the distance. Our experiments indicate that for automata appearing
in NIDSes, this simplification has typically only a negligible impact on the precision 
of the bounds.

%We employ a~greedy algorithm for guessing an optimal~$V'$---this cannot affect
%the soundness of the algorithm, it can only give a~less precise error
%estimation.)

% Using the properties of $\trim$---in particular that it removes states that
% either have an~empty language (do not accept), or have an empty banguage (are
% unreachable)---it can be shown that the set of sets of states $\{V' \in 2^Q \mid
% V \ordP V'\}$ has a~minimum element---i.e., the smallest set of states s.t.~when
% removed, all states from V will also be removed from the result.
% We denote this minimum as~$\downclosPof V$.

% \td{OL: how to compute $\errfunc$}
% \td{OL: also say what the $\errfunc$ represents}

For computing the state labelling, we provide the following three functions,
which differ in the precision they provide and the difficulty of their
computation (naturally, more precise labellings are harder to compute):
$\labfuncP^1, \labfuncP^2$, and $\labfuncP^3$.
Given an NFA~$A$ and a~PA~$P$, they generate the labellings $\stlabP^1,
\stlabP^2$, and $\stlabP^3$, respectively, defined~as
\vspace{-2.5mm}
\begin{align*}
  \stlabP^1(q) &= \sum \left\{\distrP(\blanginof{A}{q'})~\middle|~ q'\in
  \reachof{\{q\}}\cap F\right\},
  % \stlabP^1(q) &= \hspace{-7mm}\sum_{q'\in \reachinof{A}{\{q\}}\cap F}
  % \hspace{-7mm} \distrP(\blanginof{A}{q'}),
\end{align*}
\vspace{-7.5mm}
\begin{align*}
  \stlabP^2(q) &= \distrP\left(\blanginof{A}{F \cap \reachof{q}}\right),&
 \stlabP^3(q) &= \distrP\left(\blanginof{A}{q}.\langinof{A}{q}\right).\\[-7.5mm]
\end{align*}
A~state label $\stlabof q$ approximates the error of the words removed
from~$\langof A$ when $q$ is removed.
More concretely, $\stlabP^1(q)$ is a~rough estimate saying that the error can be
bounded by the sum of probabilities of the banguages of all final states
reachable from~$q$ (in the worst case, all those final states might become
unreachable).
Note that $\stlabP^1(q)$
\begin{inparaenum}[(1)]
  \item  counts the error of a~word accepted in two different final states of
    $\reachof q$ twice, and
  \item  also considers words that are accepted in some final state in
    $\reachof q$ without going through~$q$.
\end{inparaenum}
The labelling~$\stlabP^2$ deals with (1) by computing the total probability of
the banguage of the set of all final states reachable from~$q$, and the
labelling~$\stlabP^3$ in addition also deals with (2) by only considering words
that traverse through~$q$ (they can still be accepted in some final state not in
$\reachof q$ though, so even $\stlabP^3$ is still imprecise).
Note that if $A$ is unambiguous then $\stlabP^1 = \stlabP^2$.

% which is precise (wrt~$\reachinof A q$) if the banguages of final states
% are disjoint.
% % Namely in 
% % the worst case, removing of the state $q$ leads to removing also all final states 
% % reachable from $q$ and loosing all the words from banguages of these states in 
% % the reduced NFA.
% In many case, the banguages of final states are not disjoint, so $\stlabP^2$
% However, if banguages of reachable final states are not disjoint,
% we could sum the probability of some string $w$ multiple times. In the case of 
% $\stlabP^2$ we avoid of this problem, which gives us a more precise approximation 
% of the error. Note that if $A$ is unambiguous then $\stlabP^1(q) = \stlabP^2(q)$ 
% for each state $q$. The state label $\stlabP^3$ for a state $q$ refines the 
% previous state labels by considering only probabilities of words whose run in 
% $A$ constains state $q$. 

When computing the label of $q$, we first modify~$A$ to obtain~$A'$ accepting
the language related to the particular labelling. Then, we compute the value of 
$\distrP(\langof{A'})$ using the algorithm from~\S\ref{sec:prob-distance}.
Recall that this step 
is in general costly, due to the determinisation/disambiguation of~$A'$.
The key property of the labelling computation resides in the fact that if
$A$~is composed of several disjoint sub-automata, the automaton $A'$ is
typically much smaller than $A$ and thus the computation of the label is
considerable less demanding.
Since the automata appearing in regex matching for NIDS are  composed of the
union of ``tentacles'', the particular $A'$s are very small, which enables
efficient component-wise computation of the labels.

%Although the previous step can be in general quite time-demanding due to
%determinisation/disambiguation, the NFAs in our setting have the structure of
%several disjoint sub-automata, so the NFAs corresponding to the languages of
%the state label of every state can be small (wrt the size of the original NFA).

% However, in the structure of NFAs converted from REs, 
% one can often find independent subautomata\footnote{Parts of the whole automaton s.t. 
% for two states $p,r$ from two different subautomata they are not reachable, i.e., 
% $p\not\leadsto r$ and $r \not\leadsto p$.}. Therefore, the NFAs corresponding to the 
% languages of each state label can be small (compared to the size of the original 
% NFA)---they often cover only a part of one subatomaton.

The following lemma states the correctness of using the pruning reduction as an
instantiation of Algorithms~\ref{alg:general-greedy-size-reduction}
and~\ref{alg:general-greedy-error-reduction} and also the relation among
$\stlabP^1$, $\stlabP^2$, and $\stlabP^3$.\vspace*{-1mm}

\begin{restatable}{lemma}{pruningCorrect}
\label{lem:pruning-correct}
For every $x\in \{1,2,3\}$, the functions $\reduceP$, $\errfuncP$, and
$\labfuncP^x$ satisfy~\cone{}.
Moreover, consider an NFA~$A$, a~PA~$P$, and let $\stlabP^x = \labfuncP^x(A, P)$
for $x \in \{1,2,3\}$.
Then, for each $q \in \statesof A$, we have $\stlabP^1(q) \geq \stlabP^2(q) \geq
\stlabP^3(q)$.
\end{restatable}

%!!!!!!!!!!!!!!!!!!!!!!!!!!!!!!!!!!!!!!!!!!!!
\enlargethispage{6mm}
%!!!!!!!!!!!!!!!!!!!!!!!!!!!!!!!!!!!!!!!!!!!!

%*******************************************************************************
\vspace{-5.0mm}
\subsection{Self-loop Reduction}\label{sec:self-loop}
\vspace{-2mm}
%*******************************************************************************

The main idea of the self-loop reduction is to over-approximate the language
of~$A$ by adding self-loops over every symbol at selected states.
This makes some states of~$A$ redundant, allowing them to be removed without
introducing any more error.
Given an~NFA~$A = (Q, \trans, I, F)$, the self-loop reduction searches
for a~set of states $R \subseteq Q$, which will have self-loops added, and
removes other transitions leading out of these states, making some states
unreachable.
The unreachable states are then removed.

Formally, let $\selfloopof A R$ be the NFA $(Q, \trans', I, F)$ whose transition
function $\trans'$ is defined, for all $p \in Q$ and $a \in \Sigma$, as $\trans'(p,a) =
\{p\}$ if $p \in R$ and $\trans'(p,a) = \trans(p,a)$ otherwise.
%
% \begin{equation*}
%  \trans'(p,a) =
%  \begin{cases}
%    \{p\} & \mbox{if }p \in R, \\
%    \trans(p,a) & \mbox{otherwise}.
%  \end{cases}\\[-3mm]
% \end{equation*}
%
As with the pruning reduction, the natural decision problem corresponding to the self-loop reduction
is also \pspace{}-complete.\vspace*{-1mm}
\begin{restatable}{lemma}{selfloopPspace}
\label{lem:selfloop-pspace}
  Consider an NFA~$A$, a~PA~$P$, a~bound on the number of states $n \in \nat$,
  and an~error bound~$\epsilon \in \range 0 1$.
  It is \pspace{}-complete to determine whether there exists a~subset of states
  $R \subseteq \statesof A$ of the size $|R| = n$ such that $\distanP(A,
  \selfloopof A R) \leq \epsilon$.
\end{restatable}
% 
% \begin{proof}
%   \td{OL: just a sketch, will go to the appendix}
%   hardness: from universality of NFA---it holds that an NFA $A$ is universal iff $\initof
%   A \neq \emptyset$ and there exists a~set of states $R \subseteq \statesof A$
%   of the size $|Q|$ such that $\distanP(A, \selfloopof A R) \leq 0$.
% \qed
% \end{proof}

\vspace*{-1mm}The required functions in the error- and size-driven reduction algorithms are
instantiated in the following way (the subcript $\mathit{sl}$ means
\emph{self-loop}):\vspace*{-2mm}
\begin{align*}
  \reduceSLof A V = \trimof{\selfloopof A V},
  &&
  % \errfuncSLof A V \stlab = \min_{} \sum \left\{\stlabof q \mid q\in \downclosSLof
  \errfuncSLof A V \stlab = 
    \sum \left\{\stlabof q \mid q\in \min\left(\downclosSLof
    V\right)\right\},\\[-8mm]
\end{align*}
where $\downclosSLof V$ is defined in a~similar manner as~$\downclosPof V$ in
the previous section (using a~partial order $\ordSL$ defined similarly to
$\ordP$; in this case, the order $\ordSL$ has a~single minimal element, though).

% Again, we provide three different functions for computing state labellings:
The functions $\labfuncSL^1, \labfuncSL^2$, and $\labfuncSL^3$ compute the
state labellings $\stlabSL^1, \stlabSL^2$, and~$\stlabSL^3$ for an NFA~$A$ and
a~PA~$P$ defined as follows:
\vspace{-2mm}
\begin{align*}
  \stlabSL^1(q) &= \weight{P}{\blanginof{A}{q}}, &
  \stlabSL^2(q) &= \distrP\left(\blanginof{A}{q}.\Sigma^*\right),
\end{align*}
\vspace{-7mm}
\begin{align*}
  \stlabSL^3(q) &= \stlabSL^2(q) -
  \distrP\left(\blanginof{A}{q}.\langinof{A}{q}\right).\\[-6mm]
\end{align*}
Above, $\weight P w$ for a~PA~$P = (\pinit, \pfin, \{\ptransa\}_{a \in \Sigma})$
and a~word $w \in \Sigma^*$ is defined as
$\weight P w = \pinit^\top \cdot \ptrans_w \cdot \onevec$ (i.e., similarly
as~$\distrP (w)$ but with the final weights~$\pfin$
discarded), and $\weight P L$ for
$L \subseteq \Sigma^*$ is defined as $\weight P L = \sum_{w \in L} \weight P
w$.

Intuitively, the state labelling $\stlabSL^1(q)$ computes the probability
that~$q$ is reached from an initial state, so if $q$ is pumped up with all
possible word endings, this is the maximum possible error introduced by the
added word endings.
This has the following sources of imprecision:
\begin{inparaenum}[(1)]
  \item  the probability of some words may be included twice, e.g., when
    $\blanginof A q = \{a, ab\}$, the probabilities of all words from
    $\{ab\}.\Sigma^*$ are included twice in $\stlabSL^1(q)$ because
    $\{ab\}.\Sigma^* \subseteq \{a\}.\Sigma^*$, and
  \item  $\stlabSL^1(q)$ can also contain probabilities of words that are
    already accepted on a~run traversing~$q$.
\end{inparaenum}
The state labelling $\stlabSL^2$ deals with (1) by considering the probability
of the language~$\blanginof A q . \Sigma^*$, and $\stlabSL^3$ deals also with (2)
by subtracting from the result of $\stlabSL^2$ the probabilities of the words
that pass through~$q$ and are accepted.

The computation of the state labellings for the self-loop reduction is done in
a~similar way as the computation of the state labellings for the pruning
reduction (cf.~\S\ref{sec:pruning}). For a computation of $\weight P L$ one 
can use the same algorithm as for $\distrP(L)$, only the final vector for 
PA $P$ is set to $\onevec$.
%
% >>>> TV: the two paragraphs are joint to save some space
%
The correctness of Algorithms~\ref{alg:general-greedy-size-reduction}
and~\ref{alg:general-greedy-error-reduction} when instantiated using the
self-loop reduction is stated in the following lemma.\vspace*{-1mm}
\begin{restatable}{lemma}{selfLoopCorrect}
\label{lem:self-loop-correct}
For every $x\in \{1,2,3\}$, the functions $\reduceSL$, $\errfuncSL$, and
$\labfuncSL^x$ satisfy \cone{}.
Moreover, consider an NFA~$A$, a~PA~$P$, and let $\stlabSL^x = \labfuncSL^x(A,
P)$ for $x \in \{1,2,3\}$.
Then, for each $q \in \statesof A$, we have $\stlabSL^1(q) \geq \stlabSL^2(q) \geq
  \stlabSL^3(q)$.
\end{restatable}

%%%%%%%%%%%%%%%%%%%%%%%%%%%%%%%%%%%%%%%%%%%%%%%%%%%%%%%%%%%%%%%%%%%%%%%%%%%%%%%%
\vspace{-3.0mm}
\section{Reduction of NFAs in Network Intrusion Detection Systems}\label{sec:case_study}
\vspace{-2.0mm}
%%%%%%%%%%%%%%%%%%%%%%%%%%%%%%%%%%%%%%%%%%%%%%%%%%%%%%%%%%%%%%%%%%%%%%%%%%%%%%%%

We have implemented our approach in a~Python prototype named \appreal{}
(APProximate REduction of Automata and Languages)%
\footnote{\url{https://github.com/vhavlena/appreal/tree/tacas18}}
and evaluated it on the use case of network intrusion detection using
\snort{}~\cite{snort}, a~popular open source NIDS.
The version of \appreal{} used for the evaluation in the current paper is
available as an artifact~\cite{artifact} for the TACAS'18 artifact virtual
machine~\cite{tacas18-vm}.

%!!!!!!!!!!!!!!!!!!!!!!!!!!!!!!!!
\enlargethispage{5mm}
%!!!!!!!!!!!!!!!!!!!!!!!!!!!!!!!!

%*******************************************************************************
% \vspace{-2.0mm}
\subsection{Network Traffic Model}\label{sec:label}
\vspace{-1.5mm}
%*******************************************************************************

The reduction we describe in this paper is driven by a~probabilistic model
representing a~distribution over~$\Sigma^*$, and the formal guarantees
are also wrt this model.
We use \emph{learning} to obtain a model of network traffic over the 8-bit ASCII
alphabet at a~given network point.
% <<<<<<< HEAD
Our model is created from several gigabytes of network traffic
from a~measuring point of the CESNET Internet provider
% \footnote{CESNET is an Internet provider for Czech universities and research
% institutions.}
connected to a 100\,Gbps backbone link (unfortunately, we cannot provide the
traffic dump since it may contain sensitive~data).
% =======
% Our model is created from several gigabytes of network traffic that we obtained
% from a~measuring point of CESNET%
% \footnote{CESNET is an Internet provider for Czech universities and research
% institutions.}
% connected to a 100\,Gbps backbone link (the traffic dump cannot
% be made public since it may contain sensitive data).
% >>>>>>> 2dbc895d5ffe99d01241d7e6e6a0549159d2f9c8

Learning a~PA representing the network traffic faithfully is hard.
The PA cannot be too specific---although the number of different packets that
can occur is finite, it is still extremely large (a~conservative
estimate assuming the most common scenario Ethernet/IPv4/TCP would still yield
a~number over $2^{10,000}$).
If we assigned non-zero probabilities only to the packets from the dump (which
are less than $2^{20}$), the obtained model would completely ignore virtually
all packets that might appear on the network, and, moreover, the model would
also be very large (millions of states), making it difficult to use in our
algorithms.
A~generalization of the obtained traffic is therefore needed.

A~natural solution is to exploit results from the area of PA learning, such
as~\cite{Carrasco1994,Thollard2004}.
Indeed, we experimented with the use of \alergia{}~\cite{Carrasco1994},
a~learning algorithm that constructs a~PA from a~prefix tree (where edges
are labelled with multiplicities) by merging nodes that are ``similar.''
The automata that we obtained were, however, \emph{too}
general.
In particular, the constructed automata destroyed the structure of network
protocols---the merging was too permissive and the generalization merged distant
states, which introduced loops over a very large substructure in the automaton
(such a~case usually does not correspond to the design of network protocols).
As a~result, the obtained PA more or less represented the Poisson
distribution, having essentially no value for us.

In \S\ref{sec:evaluation}, we focus on the detection of malicious traffic
transmitted over HTTP.
We take advantage of this fact and create a~PA representing the traffic while
taking into account the structure of HTTP.
We start by manually creating a~DFA that represents the high-level structure of
HTTP.
Then, we proceed by feeding 34,191 HTTP packets from our sample into the DFA,
at the same time taking notes about how many times every state is reached and
how many times every transition is taken.
The resulting PA $\phttp$ (of 52 states) is then obtained from the DFA and the
labels in the obvious way.

The described method yields automata that are much
better than those obtained using \alergia{} in our experiments.
A~disadvantage of the method is that it is only semi-automatic---the basic DFA
needed to be provided by an expert.
We have yet to find an~algorithm that would suit our needs for learning more
general network traffic.

% Obtaining a~precise model of network traffic is impossible just due to the sheer
% size of it---it would need to, e.g., represent that the occurrence
% of the~sequence \texttt{user=<username>\&password=<correct password>} is more
% likely than the occurrence of the
% sequence~\texttt{user=<username>\&password=<incorrect password>}, as well as to
% represent that the occurrence of a~word (e.g. a~password) with a single typo is
% more likely than the occurrence of the same word with two typos.
% Furthermore, the model differs based on the location 

%*******************************************************************************
\vspace{-0.0mm}
\subsection{Evaluation}\label{sec:evaluation}
\vspace{-0.0mm}
%*******************************************************************************

%A possible structure:
%\begin{itemize}
%\item experimental setting + software we use + workflow of the reduction process
%\item evaluation of the precision of the error guaranties provided by our approach (includes all tables, some useful numbers are missing) 
%\item performance evaluation (includes all tables, some useful numbers are missing) 
%\item practical impact obtained by reducing the big NFA (the motivating text, and some numbers from Vojta and hardware people are missing)
%\end{itemize} 

We start this section by introducing  the experimental setting, namely, the
integration of our reduction techniques into the tool chain implementing
efficient regex matching, the concrete settings of \appreal{}, and the
evaluation environment.
% We also provide description of the tables presented in this section.
Afterwards, we discuss the results evaluating the quality of the obtained approximate reductions as well as of the provided error bounds. Finally, we present the performance of our approach and discuss its key aspects.
Due to the lack of space, we selected  the most interesting results demonstrating the potential as well as the limitations of our approach.
% \lh{strange: note that we have done more, though we will not show you}

%-------------------------------------------------------------------------------
\subparagraph{General setting.}
\snort{} detects malicious network traffic based on \emph{rules} that contain
\emph{conditions}.
The conditions may take into consideration, among others, network addresses, ports, or
Perl compatible regular expressions (PCREs) that the packet payload should
match.
In our evaluation, we always select a~subset of \snort{} rules, extract the
PCREs from them, and use~\netbench{}~\cite{Pus2011} to transform them into a~single
NFA~$A$.
Before applying $\appreal{}$, we use the state-of-the-art NFA reduction tool
\reducetool{}~\cite{Reduce} 
%from the \rabit{}~\cite{bleh} tool suite 
to decrease
the size of~$A$.
\reducetool{} performs a~language-preserving reduction of~$A$ using
%various notions of (look-ahead) simulation~\cite{Mayr2013}.
advanced variants of simulation~\cite{mayr_advanced_2013}
%\td{OL: Lukas, please check}
(in the experiment reported in Table~\ref{tab:big}, we skip the use of \reducetool{} at this step as discussed 
in the performance evaluation).
%\lh{why say here?}
The automaton $\ared$ obtained as the result of \reducetool{} is the input of
\appreal{}, which performs one of the approximate reductions from
\S\ref{sec:reductions} wrt the traffic model~$\phttp$, yielding~$\aapp$.
After the approximate reduction, we, one more time, use \reducetool{} and obtain
the result~$A'$.

%-------------------------------------------------------------------------------
\subparagraph{Settings of~\appreal{}.}
In the use case of NIDS pre-filtering, it may be important to never introduce
a~false negative, i.e., to never drop a~malicious packet.
Therefore, we focus our evaluation on the \emph{self-loop
reduction}~(\S\ref{sec:self-loop}).
In particular, we use the state labelling function $\labfuncSL^2$, since it provides 
a good trade-off between the precision and the computational demands
(recall that the computation of $\labfuncSL^2$
can exploit the ``tentacle'' structure of the NFAs we work with).
We give more attention to the \emph{size-driven reduction}~(\S\ref{sec:greedy-size})
since, in our setting, a~bound on the available FPGA resources is typically given  and 
the task is to create an~NFA with the smallest error that fits inside.
%Moreover, since we work with NFAs that are to be implemented in
%a~resource-limited FPGA, we give more attention to the \emph{size-driven
%reduction}~(\S\ref{sec:greedy-size}), with the typical setting being that
%a~bound on the available FPGA resources is given and the task is to create
%an~NFA with the smallest error that fits in.
The order $\ordalglabgrc$ over states  used in \S\ref{sec:greedy-size} and \S\ref{sec:greedy-error}
is defined as  $s \ordalglabgrc s' \Leftrightarrow  \stlabSL^2(s) \leq \stlabSL^2(s')$.
%is \emph{increasing} with respect to labels of the states.
%\td{OL: Vojta, please also check}

%-------------------------------------------------------------------------------
\subparagraph{Evaluation environment.} All experiments run on a 64-bit \linux{} \debian{} 
workstation with the Intel Core(TM) i5-661 CPU running at 3.33\,GHz with
16\,GiB of RAM.

%
%All experiments were performed on a~\gnu{}/\linux{} (\debian{} 8 Jessie 64bit)
%workstation with the Intel Core(TM) i5-661 CPU running at 3.33GHz\,GHz and
%16\,GiB of RAM.

%\td{OL: talk about the FPGA setting, too}
%FPGA processes packets, performs header extraction and payload identification
%
%
%
%
%
%\td{OL: 
%FPGA in the network probe environment of~\cite{Matousek2016}.
%Software-defined monitoring~\cite{Kekely2016}
% }

%------------------------------------------------------------------------------
\subparagraph{Description of tables.}
In the caption of every table, we provide the name of the input file (in the
directory \texttt{regexps/tacas18/} of the repository of \appreal{}) with the
selection of \snort{} regexes used in the particular experiment, together with
the type of the reduction (size- or error-driven).
All reductions are over-approximating (self-loop reduction).
We further provide the size of the input automaton~$|A|$, the size after the
initial processing by \reducetool{} ($|\ared|$), and the time of this reduction
($\timeof{\reducetool}$).
Finally, we list the times of computing the state
labelling~$\labfuncSL^2$ on~$\ared$ ($\timeof{\labfuncSL^2}$), the exact
probabilistic distance  ($\timeof{\text{Exact}}$), and also the number of
\emph{look-up tables} ($\lutsof \ared$) consumed on the targeted FPGA (Xilinx
Virtex 7 H580T) when $\ared$ was synthesized (more on this in \S\ref{sec:res}).
%\td{OL: every row is for blah}
%\td{OL: due to the lack of space, we tried to select the most interesting results}
%\lh{we can get 4 or more lines by replacing items by new line and removing "This column shows"}
The meaning of the columns in the tables is the following:
\begin{itemize}
  \item[$\boldsymbol{k}/\boldsymbol{\epsilon}$]
    is the parameter of the reduction.
    In particular, $\boldsymbol{k}$ is used for the size-driven reduction and
    denotes the desired reduction ration $k = \frac n {|\ared|}$ for an input
    NFA~$\ared$ 
    %(it is the NFA after \reducetool{})
     and the desired size of the
    output~$n$.
    On the other hand, $\boldsymbol{\epsilon}$ is the desired maximum error on
    the output for the error-driven reduction.

  \item[$\boldsymbol{|\aapp|}$]
    shows the number of states of the automaton~$\aapp$ after the
    reduction by \appreal{} and the time the reduction took (we omit it when it is not interesting).

  \item[$\boldsymbol{|A'|}$]
    contains the number of states of the NFA~$A'$ obtained after
    applying \reducetool{} on $\aapp$ and the time used by \reducetool{} at
    this step (omitted when not interesting).

  \item[\textbf{Error bound}]
    shows the estimation of the error of~$A'$ as
    determined by the reduction itself, i.e., it is the probabilistic distance  computed by the
    function $\errfunc$ in~\S\ref{sec:reductions}.

  \item[\textbf{Exact error}]
    contains the values of $\distan_{\phttp}(A, A')$ that we computed
    \emph{after} the reduction in order to evaluate the precision of the result
    given in \textbf{Error bound}.
    The computation of this value is very expensive ($time(Exact)$) since it inherently requires 
    determinisation of the whole automaton $A$.
    We do not provide it in Table~\ref{tab:big} 
    (presenting the results for the automaton $\abac$ with 1,352 states) because the
    determinisation ran out of memory
    % since we encountered an out-of-memory error
    (the step is not required in the reduction process).

  \item[\textbf{Traffic error}]
    shows the error that we obtained when compared~$A'$ with~$A$ on an HTTP
    traffic sample, in particular the ratio of packets misclassified by~$A'$ to
    the total number of packets in the sample (242,468).
    Comparing \textbf{Exact error} with \textbf{Traffic error}
    gives us a~feedback about the fidelity of the traffic model~$\phttp$. 
    We note that there are no guarantees on the relationship between
    \textbf{Exact error} and \textbf{Traffic~error}.

  \item[\textbf{LUTs}]
    is the number of LUTs consumed by $A'$ when synthesized into
    the FPGA.
    Hardware synthesis is a~costly step so we provide this value only for
    selected NFAs.
\end{itemize}

\begin{table}[t]
  \caption{Results for the
    \texttt{http-malicious} regex,
    $|\amal| = 249$,
    $|\amalred| = 98$,\\
    $\timeof{\reducetool{}} = 3.5$\,s,
    $\timeof{\labfuncSL^2} = 38.7$\,s,
    $\timeof{\text{Exact}} = {}$ 3.8--6.5\,s, and\\
    $\lutsof{\amalred} = 382$.
  }
\begin{subtable}[t]{0.49\textwidth}
\caption{size-driven reduction}
\begin{adjustbox}{center}
\setlength{\tabcolsep}{4pt}
\scalebox{0.76}{
\begin{tabular}{l|rrlllr}
\toprule
                             &                                  &                   & \mc1c{\textbf{Error}} & \mc1c{\textbf{Exact}} & \mc1c{\textbf{Traffic}}  \\[-0.8mm]
  \mc1{c|}{$\boldsymbol{k}$}   & \mc1c{$\boldsymbol{|\amalapp|}$} & \mc1c{$\boldsymbol{|\amalp|}$} & \mc1c{\textbf{bound}} & \mc1c{\textbf{error}} & \mc1c{\textbf{error}} & \mc1c{\textbf{LUTs}}   \\
 \midrule
 %0.0 & 1 (0.80\,s)  & 1 (0.2\,s)  & 1.0     & 1.0     & 1.0     & \mc1c{---} \\
 0.1 & 9 (0.65\,s)  & 9 (0.4\,s)  & 0.0704  & 0.0704  & 0.0685  & \mc1c{---} \\
 0.2 & 19 (0.66\,s) & 19 (0.5\,s) & 0.0677  & 0.0677  & 0.0648  & \mc1c{---} \\
 0.3 & 29 (0.69\,s) & 26 (0.9\,s) & 0.0279  & 0.0278  & 0.0598  & 154        \\
 0.4 & 39 (0.68\,s) & 36 (1.1\,s) & 0.0032  & 0.0032  & 0.0008  & \mc1c{---} \\
 0.5 & 49 (0.68\,s) & 44 (1.4\,s) & 2.8e-05 & 2.8e-05 & 4.1e-06 & \mc1c{---} \\
 0.6 & 58 (0.69\,s) & 49 (1.7\,s) & 8.7e-08 & 8.7e-08 & 0.0     & 224        \\
 %0.7 & 68 (0.71\,s) & 64 (2.0\,s) & 1.0e-11 & 1.0e-11 & 0.0     & \mc1c{---} \\
 0.8 & 78 (0.69\,s) & 75 (2.7\,s) & 2.4e-17 & 2.4e-17 & 0.0     & 297        \\
 %0.9 & 88 (0.69\,s) & 86 (3.1\,s) & 1.2e-25 & 5.5-22  & 0.0     & \mc1c{---} \\
% 1.0 & 98 (0.64s) & 98 (3.53s) & 0.0 & 0.0 (3.53s) & 0.0 \\ 
\bottomrule
\end{tabular}}
%
%
%    OLD    --- removed when LUTs were added
%
%
% \begin{tabular}{l|rrll@{\hspace{0.9mm}}rl}
% \toprule
%                              &                                  &                   & \mc1c{\textbf{Error}} & \mc2c{\textbf{Exact}} & \mc1c{\textbf{Traffic}}  \\[-0.8mm]
% \mc1{c|}{$\boldsymbol{k}$}   & \mc1c{$\boldsymbol{|\amalapp|}$} & \mc1c{$\boldsymbol{|\amalp|}$} & \mc1c{\textbf{bound}} & \mc2c{\textbf{error}} & \mc1c{\textbf{error}}    \\
%  \midrule
%  0.0 & 1 (0.80\,s)  & 1 (0.2\,s)  & 1.0     & 1.0     & (3.8\,s) & 1.0 \\
%  0.1 & 9 (0.65\,s)  & 9 (0.4\,s)  & 0.0704  & 0.0704  & (3.9\,s) & 0.0685 \\
%  0.2 & 19 (0.66\,s) & 19 (0.5\,s) & 0.0677  & 0.0677  & (3.9\,s) & 0.0648\\
%  0.3 & 29 (0.69\,s) & 26 (0.9\,s) & 0.0279  & 0.0278  & (4.2\,s) & 0.0598 \\
%  0.4 & 39 (0.68\,s) & 36 (1.1\,s) & 0.0032  & 0.0032  & (4.3\,s) & 0.0008 \\
%  0.5 & 49 (0.68\,s) & 44 (1.4\,s) & 2.8e-05 & 2.8e-05 & (4.6\,s) & 4.1e-06 \\
%  0.6 & 58 (0.69\,s) & 49 (1.7\,s) & 8.7e-08 & 8.7e-08 & (4.9\,s) & 0.0 \\
%  0.7 & 68 (0.71\,s) & 64 (2.0\,s) & 1.0e-11 & 1.0e-11 & (5.5\,s) & 0.0 \\
%  0.8 & 78 (0.69\,s) & 75 (2.7\,s) & 2.4e-17 & 2.4e-17 & (5.9\,s) & 0.0 \\
%  0.9 & 88 (0.69\,s) & 86 (3.1\,s) & 1.2e-25 & 5.5-22  & (6.5\,s) & 0.0 \\
% % 1.0 & 98 (0.64s) & 98 (3.53s) & 0.0 & 0.0 (3.53s) & 0.0 \\ 
% \bottomrule
% \end{tabular}}
\end{adjustbox}
\label{tab:small-k-red}
%\td{OL: I would kill \emph{at least} the row for 0.0}
\end{subtable}
\hfill
\begin{subtable}[t]{0.49\textwidth}
\caption{error-driven reduction}
\begin{adjustbox}{center}
\setlength{\tabcolsep}{4pt}
\scalebox{0.76}{
\begin{tabular}{l|rrll@{\hspace{0.9mm}}rl}
\toprule
                             &                                  &                   & \mc1c{\textbf{Error}} & \mc2c{\textbf{Exact}} & \mc1c{\textbf{Traffic}}  \\[-0.8mm]
\mc1{c|}{$\boldsymbol{\epsilon}$}   & \mc1c{$\boldsymbol{|\amalapp|}$} & \mc1c{$\boldsymbol{|\amalp|}$} & \mc1c{\textbf{bound}} & \mc2c{\textbf{error}} & \mc1c{\textbf{error}}    \\
 \midrule
 0.08  & 3  & 3  & 0.0724  & 0.0724  && 0.0720 \\
 0.07  & 4  & 4  & 0.0700  & 0.0700  && 0.0683 \\
 0.04  & 35 & 32 & 0.0267  & 0.0212  && 0.0036 \\
 0.02  & 36 & 33 & 0.0105  & 0.0096  && 0.0032 \\
 0.001 & 41 & 38 & 0.0005  & 0.0005  && 0.0003 \\
 1e-04 & 47 & 41 & 7.7e-05 & 7.7e-05 && 1.2e-05 \\
 1e-05 & 51 & 47 & 6.6e-06 & 6.6e-06 && 0.0 \\
\bottomrule
\end{tabular}}
\end{adjustbox}
\label{tab:small-e-red}
%\td{OL: no times here but should be similar to the left-hand one.  Not enough space }
\end{subtable}
\label{tab:mal}
\end{table}
%*******************************************************************************
\vspace{-2.0mm}
\subsection*{Approximation errors}
\vspace{-1.5mm}
%*******************************************************************************
%
%$\subparagraph{Results.}
%
%Table~\ref{tab:small-k-red} shows the result of the $k$-self-loop reduction for the NFA  $A_1$ describing \td{ToDo}.  $A_1$ has 249 states and \rabit{} reduces the size to 98 states (the reduction took 3.5 seconds). The second column lists,  for the reduction factor $k$, the number of states after the approximate reduction, is applied. The reduction process consists of the state labels computation that took 38.7 seconds (it  is computed only once for all $k$) and the time (second column in the brackets)  of the greedy algorithm identifying the set of states to add the self-loop. The third column shows the number of states obtained when the approximate NFA was reduced using \rabit{}. The time took by this step is again in the brackets. The forth column lists the upper bounds on the probabilisitic distance  for the particular reductions, while the fifth column  depicts the real distance computed after the reduced  automaton was found and, in the brackets, the time needed to obtain the distance. Recall that this step is not required by the reduction process and serves only for evaluating proposes. The last column shows the real error (i.e. the percentage of packets that are incorrectly accepted) made on the traffic. 

\begin{wraptable}[13]{r}{7cm}
\vspace*{-8.5mm}
  \caption{Results for the
    \texttt{http-attacks} regex,
    size-driven reduction,
    $|\aatt| = 142$,
    $|\aattred| = 112$,
    $\timeof{\reducetool{}} = 7.9$\,s,
    $\timeof{\labfuncSL^2} = 28.3$\,min,
     $\timeof{\text{Exact}} = {}$ 14.0--16.4\,min.
  }
\vspace*{-1.5mm}
\begin{adjustbox}{center}
\setlength{\tabcolsep}{4pt}
\hspace*{-4mm}
\scalebox{0.76}{
\begin{tabular}{l|rrlrlr}
\toprule
                             &                                  &                   & \mc1c{\textbf{Error}} & \mc1c{\textbf{Exact}} & \mc1c{\textbf{Traffic}}  \\[-0.8mm]
  \mc1{c|}{$\boldsymbol{k}$}   & \mc1c{$\boldsymbol{|\aattapp|}$} & \mc1c{$\boldsymbol{|\aattp|}$} & \mc1c{\textbf{bound}} & \mc1c{\textbf{error}} & \mc1c{\textbf{error}} \\
 \midrule
 0.1 & 11 (1.1s)  & 5 (0.4s)  & 1.0     & 0.9972  & 0.9957  \\
 0.2 & 22 (1.1s)  & 14 (0.6s) & 1.0     & 0.8341  &  0.2313  \\
 0.3 & 33 (1.1s)  & 24 (0.7s) & 0.081   & 0.0770   & 0.0067  \\
 0.4 & 44 (1.1s)  & 37 (1.6s) & 0.0005  & 0.0005  & 0.0010  \\
 0.5 & 56 (1.1s)  & 49 (1.2s) & 3.3e-06 & 3.3-06  & 0.0010  \\
 0.6 & 67 (1.1s)  & 61 (1.9s) & 1.2e-09 & 1.2e-09 & 8.7e-05 \\
 0.7 & 78 (1.1s)  & 72 (2.4s) & 4.8e-12 & 4.8e-12 & 1.2e-05 \\
 0.9 & 100 (1.1s) & 93 (4.7s) & 3.7e-16 & 1.1e-15 & 0.0     \\
\bottomrule
\end{tabular}}
\end{adjustbox}
\label{tab:med-k-red} 
\end{wraptable}

Table~\ref{tab:mal} presents the results of the self-loop
reduction for the NFA~$\amal$ describing regexes from \texttt{http-malicious}.
We can observe that the differences between the upper bounds on the
probabilistic distance and its real value  are negligible (typically in
the order of $10^{-4}$ or less).
We can also see that the probabilistic distance agrees with the traffic error.
This indicates a good quality of the traffic model employed in the reduction
process.
Further, we can see that our approach can provide useful trade-offs between the
reduction error and the reduction factor.
Finally, Table~\ref{tab:small-e-red} shows that a~significant reduction is
obtained when the error threshold~$\epsilon$ is increased from~0.04 to~0.07.

Table~\ref{tab:med-k-red} presents the results of the size-driven self-loop
reduction for NFA~$\aatt$ describing \texttt{http-attacks} regexes.
We can observe that the error bounds provide again a very good approximation of
the real probabilistic distance.
On the other hand, the difference between the probabilistic distance and the
traffic error is larger than for~$\amal$.
Since all experiments use the same probabilistic automaton and the same
traffic, this discrepancy is accounted to the different set of packets that
are incorrectly accepted by~$\aattred$.
If the probability of these packets is adequately captured in the traffic
model, the difference between the distance and the traffic error is small and
vice versa.
This also explains an even larger difference in Table~\ref{tab:big} (presenting
the results for $\abac$ constructed from \texttt{http-backdoor} regexes) for
% $k \in \{0.2,0.3,0.4\}$.
$k \in \range{0.2}{0.4}$.
Here, the traffic error is very small and caused by a  small set of packets
(approx.~70), whose probability is not correctly captured in the
traffic model.
Despite this problem, the results clearly show that our approach still
provides significant reductions while keeping the traffic error small:
about a 5-fold reduction is obtained for the traffic error 0.03\,\% and
a~10-fold reduction is obtained for the traffic error 6.3\,\%.
We discuss the practical impact of such a reduction in \S\ref{sec:res}.

% NO TIME
% \td{Can we run a $\epsilon$-self-loop reduction for $\epsilon \in [0.0003,0.06]$ to hopefully see better agreement between the distance and the traffic error} 

%!!!!!!!!!!!!!!!!!!!!!!
% \pagebreak
%!!!!!!!!!!!!!!!!!!!!!!

\subsection*{Performance of the approximate reduction}

\begin{wraptable}[11]{r}{7.0cm}
\vspace*{-8.5mm}
  \caption{Results for
    \texttt{http-backdoor},
    size-driven reduction,
    $|\abac| = 1,352$,
    % $\timeof{\reducetool{}} = 8.1$\,min,
    $\timeof{\labfuncSL^2} = 19.9$\,min,
    $\lutsof{\abacred} = 2,266$.
  }
  \label{tab:big}
\vspace*{-3mm}
\begin{adjustbox}{center}
\setlength{\tabcolsep}{4pt}
\hspace*{-4mm}
\scalebox{0.76}{
\begin{tabular}{l|rr@{\hspace{0.9mm}}rllr}
\toprule
                             &&                                  &                   & \mc1c{\textbf{Error}} & \mc1c{\textbf{Traffic}}  \\[-0.8mm]
  \mc1{c|}{$\boldsymbol{k}$}   & \mc1c{$\boldsymbol{|\abacapp|}$} & \mc2c{$\boldsymbol{|\abacp|}$} & \mc1c{\textbf{bound}} & \mc1c{\textbf{error}} & \mc1c{\textbf{LUTs}}    \\
 \midrule
%  0.0 & 1    (1.1\,m) & 1   & (0.2\,s)  & 1.0     & \mc1c{---} & \mc1c{---} \\
 0.1  & 135  (1.2\,m) & 8   & (2.6\,s)  & 1.0     & 0.997      &  202 \\
 0.2  & 270  (1.2\,m) & 111 & (5.2\,s)  & 0.0012  & 0.0631     &  579 \\
 0.3  & 405  (1.2\,m) & 233 & (9.8\,s)  & 3.4e-08 & 0.0003     &  894 \\
 0.4  & 540  (1.3\,m) & 351 & (21.7\,s) & 1.0e-12 & 0.0003     & 1063 \\
 0.5  & 676  (1.3\,m) & 473 & (41.8\,s) & 1.2e-17 & 0.0        & 1249 \\
 %0.6  & 811  (1.4\,m) & 594 & (1.3\,m)  & 4.6e-23 & 0.0        & 1466 \\
 0.7  & 946  (1.4\,m) & 739 & (2.1\,m)  & 8.3e-30 & 0.0        & 1735 \\
 %0.8  & 1081 (1.5\,m) & 862 & (3.4\,m)  & 3.6e-37 & 0.0        & 1886 \\
 0.9  & 1216 (1.5\,m) & 983 & (5.6\,m)  & 1.3e-52 & 0.0        & 2033 \\
% 1.0 & 1352 & 1119 & 0.0 & 0.0 \\ 
\bottomrule
\end{tabular}}
\end{adjustbox}
%(only a subset of all REs -- 49/154), $k$-self-loop reduction, pcap file: geant2-anon-http-noenc.pcap (242 468 packets)}
\end{wraptable}

In all our experiments (Tables 1--3),  we can observe that the most time-consuming
step of the reduction process is the computation of state labellings (it takes
at least 90\,\% of the total time).
The crucial observation is that the structure of the NFAs fundamentally affects
the performance of this step.
Although after \reducetool{}, the size of $\amal$ is very similar
to the size of $\aatt$, computing $\labfuncSL^2$ takes more time (28.3\,min
vs.~38.7\,s).
The key reason behind this slowdown is the determinisation (or alternatively
disambiguation) process required by the product construction underlying the
state labelling computation (cf.~\S\ref{sec:self-loop}).
For $\aatt$, the process results in a~significantly larger product when compared
to the product for~$\amal$.
The size of the product directly determines the time and space complexity of
solving the linear equation system required for computing the state labelling.

As explained in \S\ref{sec:reductions}, the computation of the state labelling
$\labfuncSL^2$ can exploit the ``tentacle'' structure of the NFAs appearing in
NIDSes and thus can be done component-wise.
On the other hand, our experiments reveal that the use of \reducetool{}
typically breaks this structure and thus the component-wise computation cannot
be effectively used.
For the NFA $\amal$, this behaviour does not have any major performance impact
as the determinisation leads to a moderate-sized automaton and the state
labelling computation takes less than~40\,s.
On the other hand, this behaviour has a dramatic effect for the NFA~$\aatt$.
By disabling  the initial application of \reducetool{} and thus preserving the
original structure of $\aatt$, we were able to speed up the state label
computation from~28.3\,min to~1.5\,min.
Note that other steps of the approximate reduction took a~similar time as
before disabling \reducetool{} and also that the trade-offs between the error
and the reduction factor were similar.
Surprisingly, disabling \reducetool{} caused that the computation of the exact
probabilistic distance became computationally infeasible because the
determinisation ran out of memory.
% \td{Can we  show the impact of disabling \reducetool{} on the precision, maybe we can show it for $\amal$?}  

Due to the size of the NFA $\abac$, the impact of disabling the initial
application of \reducetool{}  is even more fundamental.  In particular,
computing the state labelling took only 19.9\,min, in contrast to running out of
memory when the \reducetool{} is applied in the first~step (therefore, the
input automaton is not processed by \reducetool{} in Table~\ref{tab:big}; we
still give the number of LUTs of its reduced version for comparison, though).
Note that the size
of $\abac$ also slows down other reduction steps (the greedy algorithm and the
final \reducetool{} reduction).  We can, however, clearly see that computing the
state labelling is still the most time-consuming step.

\vspace{-0.0mm}
\subsection{The Real Impact in an FPGA-Accelerated NIDS}\label{sec:res}
\vspace{-0.0mm}
%*******************************************************************************
%
Further, we also evaluated some of the obtained automata in the setting
of~\cite{Matousek2016} implementing a~high-speed NIDS pre-filter.
In that setting, the amount of resources available for the regex matching
engine is 15,000 LUTs%
\footnote{We omit the analysis of flip-flop consumption because in our setting
it is dominated by the LUT consumption.}
and the frequency of the engine is~200\,MHz.
We synthesized NFAs that use a~32-bit-wide data path,
corresponding to processing 4 ASCII characters at once, which is---according to
the analysis in~\cite{Matousek2016}---the best trade-off between the utilization
of the chip resources and the maximum achievable frequency.
%
% >>>> TV: good but hard for theoreticians and not enough space
%
% (the 200\,MHz constraint satisfied with a~safe margin).
%
A~simple analysis shows that the throughput of one automaton is 6.4\,Gbps, so
in order to reach the desired link speed of 100\,Gbps, 16~units are required,
and 63~units are needed to handle 400\,Gbps.
With the given amount of LUTs, we are therefore bounded by 937 LUTs for
100\,Gbps and 238 LUTs for 400\,Gbps.

We focused on the consumption of LUTs by an implementation of the regex matching
engines for \texttt{http-backdoor}~($\abacred$) and
\texttt{http-malicious}~($\amalred$).
% We do not measure the consumed resources for all possible values of the
% parameter $\boldsymbol{k}$ in the latter case due to the computational
% complexity of this step).
%
\begin{itemize}

  \item \textbf{100\,Gbps}: For this speed, $\amalred$ can be used without
    any approximate reduction as it is small enough to fit in the available
    space.
    On the other hand, $\abacred$ without the approximate reduction is way too
    large to fit (at most 6~units fit inside the available space, yielding the
    throughput of only 38.4\,Gbps, which is unacceptable).
    The column \textbf{LUTs} in Table~\ref{tab:big} shows that
    using our framework, we are able to reduce $\abacred$ such that it uses
    894 LUTs (for $\boldsymbol{k}$ = 0.3), and so all the needed 16 units 
    fit into the FPGA, yielding the throughput over 100\,Gbps and the
    theoretical error bound of a~false positive $\leq 3.4\times
    10^{-8}$ wrt the model~$\phttp$.

  \item \textbf{400\,Gbps}: Regex matching at this speed is extremely
    challenging.
    The only reduced version of $\abacred$ that fits in the available space is
    the one for the value $\boldsymbol{k} = 0.1$ with the error bound almost 1.
    The situation is better for $\amalred$.
    In the exact version, at most 39~units can fit inside the FPGA with the
    maximum throughput of 249.6\,Gbps.
    On the other hand, when using our approximate reduction framework,
    we are able to place 63~units into the FPGA, each of the size 224 LUTs
    ($\boldsymbol{k}$~=~0.6) with the throughput over 400\,Gbps and the
    theoretical error bound of a~false positive $\leq
    8.7\times10^{-8}$ wrt the model~$\phttp$.

\end{itemize}

% We consider the above presented experimental results as highly encouraging.
% %
% In the future, we plan to investigate other approximate reductions of the NFAs,
% may be, using some variant of abstraction from abstract regular model checking
% \cite{artmc12}, adapted for the given probabilistic setting.
% %
% Another important issue for the future is to develop better ways of learning a
% suitable probabilistic model of the input traffic.

%%%%%%%%%%%%%%%%%%%%%%%%%%%%%%%%%%%%%%%%%%%%%%%%%%%%%%%%%%%%%%%%%%%%%%%%%%%%%%%%
\vspace{-2.5mm}
\section{Conclusion}\label{sec:label}
\vspace{-1.5mm}
%%%%%%%%%%%%%%%%%%%%%%%%%%%%%%%%%%%%%%%%%%%%%%%%%%%%%%%%%%%%%%%%%%%%%%%%%%%%%%%%

We have proposed a novel approach for approximate reduction of NFAs used in
network traffic filtering.
Our approach is based on a proposal of a probabilistic distance of the original
and reduced automaton using a probabilistic model of the input network traffic,
which characterizes the significance of particular packets.
We characterized the computational complexity of approximate reductions based on
the described distance and proposed a sequence of heuristics allowing one to
perform the approximate reduction in an efficient way.
Our experimental results are quite encouraging and show that we can often
achieve a very significant reduction for a negligible loss of precision.
We showed that using our approach, FPGA-accelerated network filtering on large
traffic speeds can be applied on regexes of malicious traffic where it could not
be applied before.

In the future, we plan to investigate other approximate reductions of the NFAs,
maybe using some variant of abstraction from abstract regular model checking
\cite{artmc12}, adapted for the given probabilistic setting.
Another important issue for the future is to develop better ways of learning a
suitable probabilistic model of the input traffic.

%------------------------------------------------------------------------------
\paragraph{Data Availability Statement and Acknowledgements.}
The tool used for the experimental evaluation in the current study is available
in the following figshare repository: \url{https://doi.org/10.6084/m9.figshare.5907055}.
We thank Jan Ko\v{r}enek, Vlastimil Ko\v{s}a\v{r},
and Denis Matou\v{s}ek for their help with translating regexes into automata
and synthesis of FPGA designs, and Martin \v{Z}\'{a}dn\'{i}k for providing us
with the backbone network traffic.
We thank Stefan Kiefer for helping us proving the \pspace{} part of
Lemma~\ref{lem:polytime-lang-prob} and Petr Peringer for testing our artifact.
The work on this paper was supported by
the Czech Science Foundation project 16-17538S,
the IT4IXS: IT4Innovations Excellence in Science project (LQ1602),
and the FIT BUT internal project FIT-S-17-4014.

%%%%%%%%%%%%%%%%%%%%%%%%%%%%%%%%%%%%%%%%%%%%%%%%%%%%%%%%%%%%%%%%%%%%%%%%%%%%%%%%
% \newpage
%%%%%%%%%%%%%%%%%%%%%%%%%%%%%%%%%%%%%%%%%%%%%%%%%%%%%%%%%%%%%%%%%%%%%%%%%%%%%%%%

\bibliographystyle{splncs}
\bibliography{bibliography}

\newpage
\appendix

%###############################################################################
\vspace{-0.0mm}
\section{Proofs of Lemmas}\label{app:proofs}
\vspace{-0.0mm}
%###############################################################################

Some of the proofs use the PA~$\pexp$ defined as $\pexp = \left(\onevec,
\vecof{\mu}, \left\{\vecof{\mu}_a\right\}_{a \in \Sigma}\right)$ where $\mu =
\frac{1}{|\Sigma| + 1}$.
$\pexp$ models an exponential distribution over the words from~$\Sigma^*$ (wrt
their length).
In particular, $\pexp$~assigns every word $w \in \Sigma^*$ the probability
$\distr_{\pexp}(w) = \mu^{|w| + 1}$.
We use~$\pexp$ to assign every word over~$\Sigma$ a~non-zero probability; any
other PA with the same property would work as well.

%---------------------------------------------------------------------
\polytimeLangProb*

\begin{proof}
  We prove the first and the second part of the lemma independently.
  \begin{enumerate}
    \item  \emph{Computing $\distrP(\langof A)$ is \pspace{}-complete for an
      NFA $A$.}
The membership in \pspace{} can be shown as follows.
The computation described at the end of~\S\ref{sec:prob-distance} corresponds
to solving a~linear equation system.
The system has an exponential size because of the blowup caused by the
determinisation/disambi\-guation of~$A$ required by the product construction.
The equation system can, however, be constructed by a \pspace{} transducer~$\meq$.
Moreover, as solving linear equation systems can be done using a
polylogarithmic-space transducer~$\msyslin$, one can combine these
two transducers to obtain a~\pspace{} algorithm.
Details of the construction follow:
\begin{itemize}
  \item  First, we construct a~transducer~$\meq$
    that, given an NFA~$A = (Q_A, \trans_A, I_A, F_A)$ and a~PA~$P = (\pinit,
    \pfin, \{\ptransa\}_{a \in \Sigma})$ on its input, constructs a~system of $m
    = 2^{|Q_A|}\cdot|\statesof P|$ linear equations~$\syslin(A, P)$ of~$m$
    unknowns $\unkn{R, p}$ for $R \subseteq Q_A$ and $p \in \statesof P$
    representing the product of $A'$ and $P$, where $A'$ is a~deterministic
    automaton obtained from~$A$ using the standard subset construction.
    The system~$\syslin(A, P)$ is defined as follows (cf.~\cite{Baier2016A}):
    \begin{align*}
      \unkn{R, p} =
      \begin{cases}
        \quad 0 & \hspace*{-13mm} \text{if } \langinof A R \cap \langinof{P'} p = \emptyset,\\
        \quad \!\!\sum\limits_{a \in \Sigma} ~~\sum\limits_{p' \in \statesof P}
        \!\!\!\left(\ptransaof{p}{p'} \cdot \unkn{\trans_A(R, a), p'}\right) + \pfin[p] & \text{if } R \cap F_A \neq \emptyset, \\
        \quad \!\!\sum\limits_{a \in \Sigma} ~~\sum\limits_{p' \in \statesof P}
        \!\!\!\ptransaof{p}{p'} \cdot \unkn{\trans_A(R, a), p'} & \text{otherwise},
      \end{cases}
    \end{align*}
    such that $P' = \suppof P$ and $\trans_A(R, a) = \bigcup_{r \in R}
    \trans(r,a)$.
    The test $\langinof A R \cap \langinof{P'} p = \emptyset$ can be performed
    as testing $\exists r_i \in R: \langinof A r \cap \langinof{P'} p =
    \emptyset$, which can be done in polynomial time.

    It holds that $\distrP(\langof A) = \sum_{p \in \statesof
    P}\pinit[p] \cdot \unkn{I_A,p}$.
    Although the size of $\syslin(A, P)$ (which is the output
    of~$\meq$) is exponential to the size of the input of~$\meq$, the internal
    configuration of~$\meq$ only needs to be of polynomial size, i.e., $\meq$
    works in~\pspace{}.
    Note that the size of each equation is at most polynomial.

  \item  Given a system $\syslin$ of $m$~linear equations of $m$~unknowns,
    solving~$\syslin$ can be done in the time $\bigO(\log^2 m)$ using
    $\bigO(m^k)$ number of processors for
    a~fixed~$k$~\cite[Corollary~2]{Csanky1975} (i.e., it is in the
    class~\nc{}).\footnote{%
      We use $\log k$ to denote the base-2 logarithm of~$k$.
      }
  % \item  The class NC is defined as a problem is in NC if there exist constants
  %   c and k such that it can be solved in time $\bigO(log^c n)$ using
  %   $\bigO(n^k)$ parallel processors.  Definition of NC.
  \item  According to~\cite[Lemma~1b]{Fortune1978}, a~$\bigO(\log^2 m)$
    time-bounded parallel machine can be simulated by a~$\bigO(\log^4 m)$
    space-bounded Turing machine.
    Therefore, there exists a~$\bigO(\log^4 m)$ space-bounded Turing
    machine~$\msyslin$ solving a~system of $m$~linear equations of
    $m$~unknowns.

  \item  As a~consequence, $\msyslin$ can solve $\syslin(A, P)$ using the space
    $\bigO(\log^4 (2^{|Q_A|}\cdot|\statesof P|)) = \bigO(\log^4 2^{|Q_A|}
    + \log^4 |\statesof P|)) = \bigO(|Q_A|^4 + \log^4 |\statesof P|))$.

  \item  The missing part is how to combine $\meq$ and $\msyslin$ to avoid using
    the exponential-size output tape of~$\meq$.
    For this, we use the following standard
    technique~\cite[Proposition~8.2]{Papadimitriou1994}.

    We start simulating~$\msyslin$.
    When $\msyslin$ moves its head right, we pause it and start simulating
    $\meq$ until it outputs the corresponding bit, which is directly fed
    into the input of~$\msyslin$.
    Then we pause~$\meq$ and resume the run of~$\msyslin$.
    We use a binary counter to keep the track of the position~$k$ of the head
    of~$\msyslin$ on the tape.
    When~$\msyslin$ moves the head right, we increment the value of the counter
    to $k+1$ and let~$\meq$ run until it produces another bit of the tape, when
    $\msyslin$ moves the head left, we decrement the value of the counter
    to $k-1$, restart~$\meq$, and let it run until it has produced the
    $(k-1)$-st bit of the tape.

    The internal configuration of both~$\meq$ and~$\msyslin$ is of a~polynomial
    size and the overhead of keeping track of the position of the head
    of~$\msyslin$  also requires only polynomial space.
    Therefore, the whole transducer runs in a~polynomially-bounded space.

\end{itemize}

The \pspace{}-hardness is obtained by a~reduction from the (\pspace{}-complete)
universality of NFAs: using the PA~$\pexp$ defined above, it holds that
$L(A) = \Sigma^*$ iff $\distr_{\pexp}(L(A)) = 1$.

    \item  \emph{Computing $\distrP(\langof A)$ is in \ptime{} if $A$ is a UFA.}
We modify the proof from~\cite{Baier2016A} into our setting.
%W.l.o.g.~we assume that~$A_U$ and~$P$ are trimmed. 
First, we formally define the \emph{product} of a~PA~$P = (\pinit, \pfin,
\{\ptransa\}_{a\in\Sigma})$ and an~NFA~$A = (Q, \delta, I, F)$, denoted $P
\pprod A$, as the $(|\statesof P| \cdot |Q|)$-state PPA $R = (\pinit_R, \pfin_R,
\{\ptransa^R\}_{a\in\Sigma})$
where%
\footnote{we assume an implicit bijection between states of the product $R$
and $\{1, \ldots, |Q[R]|\}$}
\begin{align*}
\pinit_R[(q_P, q_A)] = \pinit_R[q_P] \cdot |\{q_A\} \cap I|, &&&&
\pfin_R[(q_P, q_A)] = \pfin_R[q_P] \cdot |\{q_A\} \cap F|,
\end{align*}
\begin{align*}
\ptransa^R[(q_P, q_A), (q_P', q_A')] = \ptransaof{q_P}{q_P'}
\cdot |\{q_A'\} \cap \trans(q_A, a)|.
\end{align*}
Note that $R$ is not necessarily a~PA any more because for $w \in \Sigma^*$ such
that $\distrP(w) > 0$,
\begin{inparaenum}[(i)]
  \item  if $w \notin \langof A$, then $\distr_R(w) = 0$ and
  \item  if $w \in \langof A$ and $A$ can accept $w$ using $n$~different runs,
    then $\distr_R(w) = n \cdot \distrP(w)$.
\end{inparaenum}
As a~consequence, the probabilities of all words from~$\Sigma^*$ are no longer
guaranteed to add up to~1.
If $A$ is unambiguous, the second issue is avoided and $R$~preserves the
probabilities of words from $\langof A$, i.e., $\distr_R(w) = \distrP(w)$
for all
$w \in \langof A$, so $R$ can be seen as the restriction of $\distrP$ to
$\langof{A}$.
We assume~$R$ is trimmed.

In order to compute $\distrP(\langof{A})$,
% We assume a UFA $A = (Q, \delta, I, F)$ and a PA $P = (\pinit, \pfin, \{\ptransa\}_{a\in\Sigma})$.
% First, note that $\distrP(\langof{A})$ can be seen as a restriction of $\distrP$ to $\langof{A}$.
% We construct a PPA $R$ obtained by a trimed product of $P$ and $A$, i.e., $R = \trimof{P \cap A}$.
we construct a~matrix $\bE$ defined as $\sum_{a\in\Sigma} \ptrans^R_a$.
%
%
%
% a~(directed) graph $G_{P \prd A}$ whose structure corresponds to the
% product of the support of~$P$ with ~$A$ and edges are labelled by the
% probabilities from~$P$.
% Formally, $G_{P \prd A} = (V, \bE)$ s.t.~the set of vertices
% is~$V \subseteq \statesof P \times Q$ (from $P \times Q$ we remove pairs $(p,q)$ if 
% there are no both runs $p_I \runsto p \runsto p_F$ in $\supp(P)$ and $q_I \runsto q \runsto q_F$ 
% in $A$ where $p_I\in\supp(P)[I], p_F \in \supp(P)[F],q_I \in I$, and $q_F \in F$).   
% The set of edges is represented using the
% $|V|$-row square matrix\footnote{We implicitly assume a bijection between $V$ and 
% $\{1, \ldots, |V|\}$} $\bE$ defined as
%
% \begin{equation*}\label{eq:prod-matrix}
% \bE[(q_P, q_A), (q_P', q_A')] = \sum_{a\in \Sigma} \left(\ptransaof{q_P}{q_P'}
% \cdot |\{(q_A, a, q_A')\} \cap \trans| \right).
% \end{equation*}
%
Note that $\srof \bE < 1$, where $\srof \bE$ is the \emph{spectral radius}
of $\bE$.
Intuitively, $\srof \bE <~1$ holds because we trimmed the redundant states from the product of $P$ 
and $A$, the significance of paths between two nodes of length $n$ for $n \to \infty$ 
is 0 (for a full proof see~\cite{Baier2016A}).
We further use the following standard result in linear algebra:
if $\srof \bE < 1$, then
\begin{inparaenum}[(i)]
 \item  the matrix $\bI - \bE$ is invertible and
 \item  the sum of powers of $\bE$, denoted as $\bE^*$, can be computed as
   $\bE^* = \sum_{i=0}^{\infty} \bE^i = (\bI - \bE)^{-1}$~\cite{Hogben2006}
   (note that matrix inversion can be done in polynomial
   time~\cite{Solodovnikov1985}).
\end{inparaenum}

$\bE^*$ represents the reachability between nodes of
$R$, i.e., $\bE^*[r,r']$ is the sum of significances of all
(possibly infinitely many) paths from~$r$ to~$r'$ in $R$.
When related to $P$ and $A$, the matrix $\bE^*$ represents the reachability
in~$P$ wrt~$\langof A$, in particular,
\begin{equation}\label{eq:star-significance}
\bE^*[(q_P, q_A), (q_P', q_A')] =
\sum \left\{\ptrans_w[q_P, q'_P]~\middle|~ q_A \runstoover{w}{q'_A}, w \in \Sigma^*\right\}.
\end{equation}
We prove Equation~(\ref{eq:star-significance}) using the following reasoning.
First, we show that
% Equation~(\ref{eq:star-significance}) is a~consequence of the following equation, which represents the reachability of 
% words of length $n$ between the nodes of~$R$:
%
\begin{equation}\label{eq:reach-n}
  \bE^n[(q_P, q_A), (q_P', q_A')] =
\sum \left\{\ptrans_w[q_P, q'_P]~\middle|~ q_A \runstoover{w}{q'_A}, w \in \Sigma^n\right\},
\end{equation}
i.e., $\bE^n$ represents the reachability in $P$ wrt $\langof A$ for words of length $n$.
We prove Equation~\eqref{eq:reach-n} by induction on $n$: for $n = 0$ and $n =
1$ the equation follows directly from the definition of~$R$ and $\ptrans$.
Next, suppose that~\eqref{eq:reach-n} holds for $n > 1$;
we show that it holds also for $n+1$.
We start with the following reasoning:
\begin{align*}
  \bE^{n+1}&[(q_P, q_A), (q_P', q_A')] \\
  {}={}&(\bE^n \bE)[(q_P, q_A), (q_P', q_A')] \\
  {}={}&\sum \Big\{\bE^{n}[(q_P, q_A), (q_P'', q_A'')] \cdot
  \bE[(q_P'', q_A''), (q_P', q_A')]~\Big|~ (q_P'', q_A'')\in\statesof R\Big\}.
\end{align*}
Note that the last line is obtained via definition of matrix multiplication.
Further, using the induction hypothesis, we get
\begin{align*}
  \bE^{n+1}&[(q_P, q_A), (q_P', q_A')] \\ 
  {}={}& \sum \Bigg\{\sum \Big\{\ptrans_w[q_P, q''_P]~\Big|~ q_A \runstoover{w}{q''_A}, w \in \Sigma^n\Big\} \cdot {} \\ 
  &{}\cdot\sum \Big\{\ptrans_a[q_P'', q'_P]~\Big|~ q_A'' \ltr{a}{q'_A}, a \in \Sigma\Big\}~\Bigg|~ (q_P'', q_A'')\in\statesof R\Bigg\} \\
  {}={}&\sum \Bigg\{\sum \Big\{\ptrans_w[q_P, q''_P]\cdot\ptrans_a[q_P'', q'_P]~\Big|~ q_A \runstoover{w}{q''_A}, q_A'' \ltr{a}{q'_A}, \\ 
  &\hspace*{19mm} a \in \Sigma, w \in \Sigma^n\Big\}~\Bigg|~ (q_P'', q_A'')\in\statesof R\Bigg\} \\
  {}={}& \sum \Big\{\ptrans_{w'}[q_P, q'_P]~\Big|~ q_A \runstoover{w'}{q'_A}, w' \in \Sigma^{n+1}\Big\}.
\end{align*}
% 
% We use~$\bE^*$ to compute the matrix~$\bF^*$ that can be seen as
% mapping~$\bE^*$ back to the automaton~$A$ \vh{P ??}
% 
% \begin{equation}
% \bF^*[q_P, q_P'] =
% \sum \left\{\bE^*[(q_P, q_A), (q'_P, q'_A)] ~\middle|~ q_A, q'_A \in
%  Q\right\}.
% \end{equation}
% %
% The matrix $\bF^*$ therefore represents the reachability in $P$ wrt the
% language of~$A$:
% \begin{equation}
% \bF^*[q_P, q_P'] = \sum \left\{p ~\middle|~ q_P \runstooverp{w}{p} q'_P, w \in \langof
%  {A} \right\}.
% \end{equation}
%
Since $\bE^* = \sum_{i=0}^\infty \bE^i$, Equation~\eqref{eq:star-significance}
follows.
Using the matrix~$\bE^*$, it remains to compute $\distrP(\langof{A})$ as
\begin{equation*}\label{eq:pr-comp-final}
  \distrP(\langof{A}) = \pinit^\top_R \cdot \bE^* \cdot \pfin_R.
\end{equation*}
\qed
% where $\pinit_R$ is a vector of initial weights of $R$ and $\pfin_R$ is a vector of final weights of~$R$.
% 
% \td{OL: a bit messy now, improvement is needed!}
% 
% \td{OL: justify $\srof \bE < 1$}
% 
% \vh{Necessary?}
% \begin{equation*}
% E = \left\{((q_P, q_A), k, (q_P', q_A')) \middle|
% k = \sum_{a\in \statesof{P}} \{\ptransaof{q_P}{q_P'} \cdot |\{(q_A, a, q_A')\} \cap \trans|\}, k > 0\right\}.
% \end{equation*}
% 
% \td{OL: heh, a little ugly... :-)}
  \end{enumerate}
\end{proof}

% Lemma~\ref{lem:prob-comp} gives us a straightforward algorithm for computing 
% $\distrP(L(A))$ for a PA $P$ and an NFA $A$. In the first step, we check, 
% whether $A$ is unambiguous. This can be done by a construction of NFA 
% $B = \trim(A \cap A)$ and checking if there is some state $(p,q)\in Q[B]$ 
% s.t. $p \neq q$ \vh{Citation needed.}. If $A$ is not unambiguous, we perform 
% determinization (or disambiguation \vh{Citation needed}). In the next step, 
% according to Equation~\ref{eq:prod-matrix}, we construct a matrix $\bE$ 
% and compute the matrices $\bE^*$ and $\bF^*$. Finally, using 
% Equation~\ref{eq:pr-comp-final} we obtain the result.
% \vh{Maybe this paragraph is redundant (high-level view in Section 3)}

%-------------------------------------------------------------------------
\generalPspace*

\begin{proof}
Membership in \pspace{}: We non-deterministically
generate an automaton~$A'$ with $n$~states and test (in \pspace{}, as shown in
Lemma~\ref{lem:polytime-lang-prob}) that $\distanP(A, A') \leq \epsilon$.
This shows the problem is in \npspace{} $=$ \pspace{}.

\pspace{}-hardness: We use a reduction from the problem of checking universality
of an~NFA $A = (Q, \trans, I, F)$ over $\Sigma$, i.e., from checking
whether $\langof A = \Sigma^*$, which is \pspace{}-complete.
%~\cite{Meyer1972}.
%
First, for a reason that will become clear below, we test whether $A$ accepts
all words over $\Sigma$ of length~0 and 1, which can be done in polynomial time.
We can now show that $\langof A = \Sigma^*$ iff there is a~1-state NFA~$A'$
s.t.~$\distanPexp(A, A') \leq 0$ where $\pexp$ is as defined at the beginning of
Appendix.
The implication from left to right is clear: such an NFA can only be (up to
  isomorphism) the NFA $A' = (\{q\}, \{q \ltr{a} q \mid a \in \Sigma \}
  ,\{q\},\{q\})$).
To show the reverse implication, we note that we have tested that $\{\epsilon\}
  \cup \Sigma \subseteq \langof A$.
Since the probability of any word from $\{\epsilon\} \cup \Sigma
\subseteq \langof A$ in $\pexp$ is non-zero, the only 1-state NFA that processes
those words with zero error is the NFA~$A'$ defined above.
Because the language of $A'$ is $\langof{A'} = \Sigma^*$, it holds that
$\distanPexp(A, A') \leq 0$ iff $\langof A = \Sigma^*$.
%
% Because $A'$ is limited to have only one state $q$, in order to accept \td{OL: }
% %
% We have tested that $\{\varepsilon\}\cup\Sigma \subseteq L(A)$.
% %
% Therefore, if there is a 1-state automaton $A'$ s.t.  $\distanP(A, A') \leq 0$,
% then $A'$ must also accept all words over $\Sigma$ of length~0 and 1, and hence
% $L(A') = \Sigma^*$.
% %
% This, together with the fact that $P$ assigns a nonzero probability to each
% word from $\Sigma^*$, implies $L(A) = \Sigma^*$.
%
\qed
\end{proof}

\pruningPspace*

\begin{proof}
Membership in \pspace{}: We non-deterministically generate a subset $R$ 
of $Q[A]$ having $n$~states and test (in \pspace{}, as shown in
Lemma~\ref{lem:polytime-lang-prob}) that $\distanP(A, \restr A R) \leq \epsilon$.
This shows the problem is in \npspace{} $=$ \pspace{}.

\pspace{}-hardness:
We reduce from the \pspace{}-complete problem of checking universality of an
NFA~$A = (Q, \trans, I, F)$ over~$\Sigma$.
Consider a~symbol $x\notin\Sigma$.
Let us construct an NFA $A'$ over $\Sigma \cup \{x\}$ s.t.
$L(A') = x^*.L(A)$.
$A'$ is constructed by adding a~fresh state $\qnew$ to~$A$ that can
loop over~$x$ and make a~transition to any initial state of~$A$ over~$x$:
$A' = (Q \dunion \{\qnew\}, \trans \cup \{\qnew
\ltr{x} q \mid q \in I \cup \{\qnew\}\}, I \cup \{\qnew\}, F)$.
We set $n = |A'| + 1$.
Further, we also construct an $(n+1)$-state NFA $B$ accepting the 
language $x^n.\Sigma^*$ defined as $B = (Q_B, \trans_B, \{q_1\}, \{q_{n + 1}\})$ 
where $Q_B = \{q_1, \ldots, q_{n+1}\}$ and $\trans_B = \{q_i \ltr{x} q_{i+1} \mid 1 
\leq i \leq n\} \cup \{q_{n+1} \ltr{a} q_{n+1} \mid a \in \Sigma\}$.
Moreover, let $P$ be a~PA representing
a~distribution $\distrP$ that is defined for each $w\in(\Sigma\cup\{x\})^*$ as
\begin{equation}
 \distrP(w) = \begin{cases}
   \mu^{|w'| + 1} & \text{for~} w = x^n.w', w'\in\Sigma^*, \text{and~} \mu = \frac{1}{|\Sigma| + 1}, \\
   0 & \mbox{otherwise}.
 \end{cases}
\end{equation}
Note that $\distrP(x^n . w) = \distr_{\pexp}(w)$ for $w \in \Sigma^*$ and
$\distrP(u) = 0$ for $u \notin x^n.\Sigma^*$ ($P$ can be easily constructed
from $\pexp$).
Also note that $B$ accepts exactly those word~$w$ such that $\distrP(w) \neq
0$ and that $\distrP(\langof B) = 1$.
Using the automata defined above, we construct an NFA $C = A' \cup B$
where the union of
two NFAs is defined as $A_1 \cup A_2 = (\statesof {A_1} \dunion \statesof
{A_2}, \transof{A_1} \dunion \transof{A_2},$    %%%%%%% broken line
$\initof{A_1} \dunion \initof{A_2}, \finof{A_1} \dunion \finof{A_2})$.
The language of $C$ is $\langof C  = x^*. \langof A \cup x^n . \Sigma^*$ and its
probability is $\distrP(\langof C) = 1$.

The important property of~$C$ is that if there exists a~set 
$R \subseteq Q[C]$ of the size $|R| = n$ s.t. $\distanP(C, \restr C R) \leq 0$, 
then $L(A) = \Sigma^*$.
The property holds because
since $|Q[A']| = n-1$, when we remove $n$ states from~$C$, at least one state
from $Q[B]$ is removed, making the whole subautomaton of~$C$ corresponding
to~$B$ useless, and, therefore, $\langof{\restr C R} \subseteq x^* . \langof
A$.
Because $\distanP(C, \restr C R) \leq 0$, we know that $\distrP(\langof{\restr C
R}) = 1$, so $x^n . \Sigma^* \subseteq x^* . \langof A  = \langof{\restr C R}$
and, therefore, $\langof A = \Sigma^*$.
For the other direction, if $\langof A = \Sigma^*$, then there exists a set $R
\subseteq \statesof{A'} \cup \statesof B$ of the size $|R| = n$
s.t.~$\distanP(C, \restr C R) \leq 0$ (in particular, $R$ can be such that $R
\subseteq \statesof B$).
\qed
\end{proof}

%-----------------------------------------------------------------------
\pruningCorrect*

\begin{proof}
First (since we will use it in the proof of the first part), for each $q\in
Q[A]$ we prove the inequalities $\stlabP^1(q) \geq \stlabP^2(q) \geq
\stlabP^3(q)$. 
The first inequality follows from the fact that if the banguages of reachable
final states are not disjoint, in the case of $\stlabP^1$ we may sum
probabilities of the same words multiple times.
The second inequality follows from the inclusion
$\blanginof{A}{q}.\langinof{A}{q} \subseteq \blanginof{A}{F \cap
\reachof{q}}$.

Second, we prove that the functions $\reduceP$, $\errfuncP$, and
$\labfuncP^x$ satisfy the properties of~\cone{}:
\begin{itemize}
  \item \cone{}(a): To show that $\errfuncPof A V {\labfuncP^x(A, P)} \geq
    \distanP(A, \reducePof{A}{V})$, we prove that the inequality holds
    for~$\stlabP^3 = \labfuncP^3(A, P)$; the rest follows from $\stlabP^1(q)
    \geq \stlabP^2(q) \geq \stlabP^3(q)$ proved above.

    Consider some set of states $V \subseteq
    Q[A]$ and the set $V'\in \downclosPof V$ s.t.~for any $V'' \in \downclosPof
    V$, it holds that $\sum\{ \stlabP^3(q) \mid q \in V'\} \leq \sum\{
    \stlabP^3(q) \mid q \in V''\}$.
    We have
    \begin{equation}\label{eq:pr-err-subset}
     \begin{aligned}
       \langof A \symdiff {} &\langof{\reducePof{A}{V}} \\
       &=\langof A \symdiff \langof{\reducePof{A}{V'}} &
       \expl{def.~of~$\ordP$}\\ 
       &= \langof A \setminus \langof{\reducePof{A}{V'}} &
       \expl{$\langof A \supseteq \langof{\reducePof{A}{V'}}$}\\ 
       &\subseteq \bigcup_{q \in V'}\blanginof{A}{q}.\langinof{A}{q}.
       &\expl{def.~of $\reduceP$}
     \end{aligned}
    \end{equation}
    Finally, using~(\ref{eq:pr-err-subset}), we obtain
    \begin{equation}
      \begin{aligned}
       \distanP(A,&\reducePof{A}{V}) \\
        {}={}& \distrP(L(A) \setminus L(\reducePof{A}{V'})) \\
        &\expl{def.~of~$\distanP$ and $\langof A \supseteq \langof{\reducePof{A}{V'}}$}\\ 
        {}\leq{}& \sum_{q \in V'}\distrP(\blanginof{A}{q}.\langinof{A}{q}) &
        \expl{(\ref{eq:pr-err-subset})}\\
        {}={}& \sum\{ \stlabP^3(q) \mid q \in V'\} &
        \expl{def.~of~$\stlabP^3$}\\
        {}={}& \hspace*{-3mm}\min_{V'' \in \downclosPof V}\hspace*{-1mm}
        \sum\{ \stlabP^3(q) \mid q \in V''\} &
        \expl{def.~of~$V'$}\\
        {}={}& \errfuncPof{A}{V}{\stlabP^3}.
        &\expl{def.~of $\errfuncP$}
      \end{aligned}
    \end{equation}
  \item \cone{}(b): $|\reducePof A {Q[A]}| \leq 1$ because
    $|\reducePof{A}{Q[A]}| = |\trimof{\restr A \emptyset}| = 0$.
  \item \cone{}(c): $\reducePof A {\emptyset} = A$ since
    $\reducePof{A}{\emptyset} = \trimof{\restr A {\statesof A}} = A$ (we assume
    that $A$~is trimmed at the input).
    \qed
\end{itemize}
\end{proof}

%-------------------------------------------------------------------------------
\selfloopPspace*

\begin{proof}
Membership in \pspace{} can be proved in the same way as in the proof of
Lemma~\ref{lem:pruning-pspace}.

\pspace{}-hardness:
We reduce from the \pspace{}-complete problem of checking universality of an
NFA~$A = (Q, \trans, I, F)$.
First, we check whether $I[A] \neq \emptyset$.
% Next, we consider a one-state PA $P$ that assigns a nonzero 
% probability to each word from $\Sigma^*$, e.g., $P = \left(\onevec, \vecof{\mu}, 
% \left\{\vecof{\mu}_a\right\}_{a \in \Sigma}\right)$ where $\mu = \frac{1}{|\Sigma| + 1}$.
We have that $\langof A = \Sigma^*$ iff there exists a set of states 
$R \subseteq Q$ of the size $|R| = |Q|$ such that 
$\distanPexp(A, \selfloopof A R) \leq 0$ (note that this means that a~self-loop
is added to every state of~$A$).
\qed
\end{proof}

%-------------------------------------------------------------------------------
\selfLoopCorrect*

\begin{proof}
First, because we use them in the proof of the first part, for each $q\in Q[A]$
we prove the inequalities $\stlabSL^1(q) \geq \stlabSL^2(q) \geq \stlabSL^3(q)$.
We start with the equality $\weight{P}{w} = \distrP(w.\Sigma^*)$ and 
obtain the equality
\begin{equation}
  \sum_{w \in \blanginof{A}{q}}\weight{P}{w} = \sum_{w \in \blanginof{A}{q}}\distrP(w.\Sigma^*),
\end{equation}
which, in turn, implies $\stlabSL^1 = \weight{P}{\blanginof{A}{q}} \geq
\distrP\left(\blanginof{A}{q}.\Sigma^*\right) = \stlabSL^2$.
The previous holds because if, e.g., $\blanginof{A}{q} = \{w, wa\}$ for
$w\in\Sigma^*$ and $a \in \Sigma$, then
$$
\weight P {\blanginof{A}{q}} = \weight P {\{w, wa\}} = \weight P w + \weight
P {wa},
$$
while
$$
\distrP\left(\blanginof{A}{q}.\Sigma^*\right) = \distrP\left(\{w,
wa\}.\Sigma^*\right) = \distrP\left(w.\Sigma^*\right).
$$
The inequality~$\stlabSL^2 \geq \stlabSL^3$ holds trivially.

Second, we prove that the functions $\reduceSL$, $\errfuncSL$, and
$\labfuncSL^x$ satisfy the properties of~\cone{}:
\begin{itemize}
  \item \cone{}(a): To show that $\errfuncSLof A V {\labfuncSL^x(A, P)} \geq
    \distanP(A, \reduceSLof{A}{V})$, we prove that the inequality holds
    for~$\stlabSL^3 = \labfuncSL^3(A, P)$; the rest follows from $\stlabSL^1(q)
    \geq \stlabSL^2(q) \geq \stlabSL^3(q)$ proved above.

    Consider some set of states $V \subseteq Q[A]$ and the set $V' =
    \min(\downclosSLof V)$.
    We can estimate the symmetric difference of the languages of the
    original and the reduced automaton as
    \begin{equation}\label{eq:sl-err-lang}
     \begin{aligned}
       \langof A \symdiff {} &\langof{\reduceSLof{A}{V}}  \\
       {}={}& \langof A \symdiff \langof{\reduceSLof{A}{V'}}
       &\expl{def.~of~$\ordSL$} \\
       {}={}& \langof{\reduceSLof{A}{V'}} \setminus \langof A 
       &\hspace*{-5mm}\expl{$\langof A \subseteq \langof{\reduceSLof{A}{V'}}$} \\
       {}\subseteq{}& \bigcup_{q \in V'}\blanginof{A}{q}.\Sigma^* \setminus
       \bigcup_{q \in V'}\blanginof{A}{q}.\langinof{A}{q}.
       &\expl{def.~of $\reduceSL$}
     \end{aligned}
    \end{equation}
    The last inclusion holds because~$\selfloopof A V$ adds self-loops to the
    states in~$V$, so the newly accepted words are for sure those that traverse
    through~$V$, and they are for sure not those that could be accepted by
    going through~$V$ before the reduction (but they could be accepted without
    touching~$V$, hence the inclusion).
    We can estimate the probabilistic distance of $A$ and $\reduceSLof{A}{V}$ as
    \begin{equation}
    \begin{aligned}
      \distanP(A,& \reduceSLof{A}{V}) \\
      {}\leq{}& \distrP\bigg(\bigcup_{q \in V'}\blanginof{A}{q}.\Sigma^*
      \setminus \bigcup_{q \in V'}\blanginof{A}{q}.\langinof{A}{q}\bigg)
      &\expl{(\ref{eq:sl-err-lang})}\\
      {}={}&\distrP\bigg(\bigcup_{q \in V'}\blanginof{A}{q}.\Sigma^*\bigg) - 
      \distrP\bigg(\bigcup_{q \in V'}\blanginof{A}{q}.\langinof{A}{q}\bigg)\\
      & \expl{property of $\distr$}\\
      {}\leq{}&
      \sum_{q \in V'}\distrP\left( \blanginof{A}{q}.\Sigma^* \right) - 
      \sum_{q \in V'}\distrP\left( \blanginof{A}{q}.\langinof{A}{q} \right) \\
      & \expl{property of $\distr$ and the fact that $\blanginof{A}{q}.\Sigma^*
      \subseteq \blanginof{A}{q}.\langinof{A}{q}$}\\
      {}={}&
      \sum_{q \in V'}\left(\distrP\left( \blanginof{A}{q}.\Sigma^* \right) - 
      \distrP\left( \blanginof{A}{q}.\langinof{A}{q} \right)\right) \\
      {}={}&
      \sum_{q \in V'}\left(\distrP\left( \blanginof{A}{q}.\Sigma^* \right) - 
      \distrP\left( \blanginof{A}{q}.\langinof{A}{q} \right)\right) \\
      {}={}&
      \sum \{\stlabSL^3(q) \mid q \in \min(\downclosSLof{V})\}\\
      &\expl{def.~of $\stlabSL^3$ and $V'$}\\
      {}={}&\errfuncSLof{A}{V}{\stlabSL^3}.
    \end{aligned}
    \end{equation}

  \item \cone{}(b): $|\reduceSLof A {Q[A]}| \leq 1$ because, from the
    definition, $|\reduceSLof{A}{Q[A]}| = |\trimof{\selfloopof A Q[A]}|
    = 1$.
  \item \cone{}(c): $\reduceSLof A {\emptyset} = A$ since
    $\reduceSLof{A}{\emptyset} = \trimof{\selfloopof A \emptyset} = A$ (we
    assume that $A$~is trimmed at the input).\qed
\end{itemize}
\end{proof}

% \subsection{Computation of State Labels}
% \vh{Strange section}
% Let us consider a PA $P = (\pinit, \pfin, \{\ptransa\}_{a\in\Sigma})$ and a UFA $A$.
% For an improved computation of $\distrP(\blanginof A q)$, one can construct
% vectors
% \begin{align}
%   \pinit^\flat[(q_P, q_A)] &= \begin{cases}
%                               \pinit[q_P] & q_A \in I[A], \\
%                               0 & \mbox{otherwise},
%                              \end{cases}\\
%   \pfin^\flat_q[(q_P, q_A)] &= \begin{cases}
%                                \pfin[q_P] & q_A = q, \\
%                                0 & \mbox{otherwise}.
%                              \end{cases}
% \end{align}
% Then, if we consider a computation of matrix $\bE^*$ from Lemma~\ref{lem:prob-comp}, we have
% \begin{equation}
%   \distrP(\blanginof A q) = (\pinit^\flat)^\top\cdot\bE^*\cdot\pfin^\flat_q.
% \end{equation}
% The advantage of this approach is that we can compute the inverse of matrix $\bE$ 
% only once and $\distrP(\blanginof A q)$ we can get for each $q \in Q[A]$ only by modifying the vector 
% $\pfin^\flat_q$.

\end{document}